\documentclass{qam-l}

\usepackage{amsfonts,latexsym}
\usepackage{stmaryrd}
\usepackage{enumitem}
\usepackage{amsmath,amssymb}
\usepackage{xcolor}

\usepackage[normalem]{ulem}
\usepackage{cancel}

\newtheorem{theorem}{Theorem}[section]
\newtheorem{lemma}[theorem]{Lemma}

\newtheorem{corollary}[theorem]{Corollary}

\theoremstyle{definition}
\theoremstyle{definition}

\newtheorem{assumption}{Assumption}[section]
\newtheorem{problem}{Problem}[section]

\theoremstyle{remark}
\newtheorem{remark}{Remark}[section]

\numberwithin{equation}{section}

\newcommand{\dd}{\,\mathrm{d}}
\newcommand{\mathbi}[1]{{\boldsymbol{#1}}}

\newcommand{\etal}{\textit{et al.}\ }
\newcommand{\cf}{\textit{cf.\ }}
\newcommand{\eg}{\textit{e.g.\ }}

\newcommand{\mm}{{-}}
\newcommand{\pp}{{+}}
\newcommand{\br}[1]{{(#1)}}

\newcommand{\std}{{\mathrm{ss}}}

\newcommand{\frc}{{\mathrm{f}}}
\newcommand{\frcpm}{{\frc_\mm^\pp}}

\newcommand{\norm}[1]{{\big\Vert #1 \big\Vert}}

\newcommand{\innp}[1]{{\big( #1 \big)}}

\newcommand{\jmp}[1]{{\,\big[\,#1\,\big]_\mm^\pp }}

\newcommand{\DDt}[1]{\frac{\dd #1}{\dd t}}
\newcommand{\ddt} [1]{\dot {#1}}
\newcommand{\ddtt}[1]{\ddot{#1}}
\newcommand{\pdd}[2]{\frac{\partial #1}{\partial #2}}
\newcommand{\pddt}{\frac{\dd}{\dd t}}

\newcommand{\Dt}{{\delta t\,}}

\newcommand{\TT}{\parallel}
\newcommand{\nablaS}{\nabla^{^\Surf}\!}

\newcommand{\Domn}{{\Omega}}

\newcommand{\Surf}{{\Sigma}}

\newcommand{\fSurf}{\Surf_\mathrm{f}}

\newcommand{\itLa}{\mathit{\Lambda}}
\newcommand{\itLaINI}{\itLa^{\tsT^0}}

\newcommand{\tsC}{\mathbi{C}}
\newcommand{\tsLa}{\mathbi{\mathit{\Lambda}}}
\newcommand{\tsLaINI}{\tsLa^{\!\tsT^0\!}}
\newcommand{\tsGa}{\mathbi{\Gamma}}
\newcommand{\tsGaINI}{\tsGa^{\tsT^0}}

\newcommand{\tsT}{\mathbi{T}}

\newcommand{\tsTd}{\mathbi{T}_\delta}

\newcommand{\vef}{\mathbi{f}}
\newcommand{\veg}{\mathbi{g}}
\newcommand{\vex}{\mathbi{x}}

\newcommand{\veu}{\mathbi{u}}
\newcommand{\ven}{\mathbi{n}}
\newcommand{\vev}{\mathbi{v}}
\newcommand{\ves}{{{\scriptsize\mathbi{\mathcal{S}}}}}
\newcommand{\vew}{\mathbi{w}}

\newcommand{\veer}{\mathbi{\epsilon}}

\newcommand{\scer}{{\epsilon}}
\newcommand{\Veer}{\hat{\mathbi{\eta}}}
\newcommand{\Scer}{{\hat\eta}}

\newcommand{\vetau}{{\mathbi{\tau}}}

\newcommand{\vei}{\mathbi{i}}
\newcommand{\vej}{\mathbi{j}}

\newcommand{\veps}{\varepsilon}
\newcommand{\scalF}{F^\veps}
\newcommand{\bcalF}{\mathbi{F}^\veps}
\newcommand{\sPsi}{\Psi^\veps}
\newcommand{\bPsi}{\mathbi{D}^\veps}

\newcommand{\opF}{\mathcal{F}}
\newcommand{\opG}{\mathcal{G}}
\newcommand{\opS}{\mathcal{S}}
\newcommand{\Sym}{\mathfrak{S}}

\newcommand{\spH}{\mathcal{H}}

\newcommand{\Grav}{G}

\definecolor{Blue}{rgb}{0,0,0.7}
\definecolor{Red}{rgb}{0.7,0,0}

\newcommand{\Cblu}{\color{black}}

\begin{document}

\title
[Dynamic ruptures generating seismic waves
in a self-gravitating planet]
{
    Analysis of dynamic ruptures generating seismic waves in a
    self-gravitating planet: An iterative coupling scheme
    and well-posedness
}

\author{
    Maarten V. de Hoop 
}
\address{Simons Chair in
    Computational and Applied Mathematics and Earth Science,
    Rice University, 6100 Main Street,
    Houston, Texas 77005, USA}

\author{
    Kundan Kumar 
}
\address{Department
    of Mathematics,
    University of Bergen, Allegaten 41,
    Postboks 7803, 5020 Bergen, Norway}

\author{
    Ruichao Ye 
}
\address{Department
    of Earth, Environmental and Planetary Sciences, 
    Rice University, 6100 Main Street,
    Houston, Texas 77005, USA; now at
    ExxonMobil Upstream Research Company,
    22777 Springwoods Parkway, Spring, Texas 77389, USA}
\email{ruichao.ye@gmail.com}


\subjclass[2000]{Primary }
\date{}
\dedicatory{}

\begin{abstract}
We study the solution of the system of equations describing the
dynamical evolution of spontaneous ruptures generated in a prestressed
elastic-gravitational deforming body and governed by rate and state
friction laws. We propose an iterative coupling scheme based on a weak
formulation with nonlinear interior boundary conditions, both for
continuous time and with implicit discretization (backward Euler) in
time. We regularize the problem by introducing viscosity. This
guarantees the convergence of the scheme for solutions of the
regularized problems in both cases. We also make precise the
conditions on the relevant coefficients for convergence to hold.
\end{abstract}

\maketitle
\allowdisplaybreaks

\section{Introduction}

The study and mathematical formulation of seismic wave propagation and
scattering in a uniformly rotating and self-gravitating earth model
dates back to the works of Dahlen \cite{Dahlen1972,Dahlen1973} and
Woodhouse and Dahlen \cite{Woodhouse1978}. Valette \cite{Valette1987}
studied the proper weak formulation of the underlying system of
equations, and de Hoop, Holman and Pham \cite{deHoop2015} completed
the analysis of well-posedness through energy estimates. The
complications in this analysis arise essentially from the presence of
a fluid outer core, and Shi \etal \cite{Shi2018} showed its impact
upon the interior normal modes of the earth while solving a generalized
eigenvalue problem.  Here, we study a different complication, namely
the coupling of the system to rupture dynamics.

The energy budget of a ``kinematic'' rupture via a slip boundary
condition was studied by Dahlen \cite{Dahlen1977}, without a friction
law. However, in rupture dynamics friction laws play a critical
role. Theoretical models of earthquake rupturing based on rate and
state friction laws and their incorporation in the
elastic-gravitational system of equations describing seismic waves
have been studied in recent years \cite{Lozos2015, Harris2017,
  Thomas2017}. However, a rigorous mathematical, weak formulation and
an analysis of well-posedness have been lacking and are addressed here
while introducing an iterative coupling scheme.

The dependence of friction strength on slip rate and the evolving
contact properties of material, or so-called ``state'', has been
recognized in laboratory studies and formalized by Dieterich
\cite{Dieterich1979}, Ruina \cite{Ruina1981,Ruina1983}, Rice
\cite{Rice1983}, Rice and Ruina \cite{Rice1983a}, and many
others. Such studies were conducted on various rock types and fault
gouge layers, and over a wide range of slip rates and confining normal
stresses. The relation between the rate and state friction laws and
realistic rupture processes was discussed by Dunham \etal
\cite{Dunham2011}.

Originally developed in the laboratory, rate- and state-dependent
friction laws have been proven to be well-posed in one-dimensional
problems; these laws are known to fit rate-dependent experimental
results \cite{Dieterich1979,Ruina1983,Richardson1999,Rice2001}.
However, general existence or uniqueness of solutions for rate- and
state-dependent friction laws coupled to the elastic-gravitational
system of equations (in dimension three) describing oscillations of
the earth have not been studied so far. The challenges that arise in
proving these pertain to the high-order derivative terms arising from
the dependency of friction on the normal stress and the surface
divergence introduced by a dynamically slipping boundary. The
challenges remain in the coupling of slip-dependent friction laws,
even for the simplest case, that is, linear slip-weakening friction.
Analyses of well-posedness have been based on simplified scenarios
deviating from the general case in essential ways: Fixing the normal
stress to a reference value (the Tresca model \cite{Ionescu1996,
  Ionescu2003, Pipping2015, Pipping2017}), or characterizing the
normal stress with a power-relation of normal displacement (the normal
compliance model \cite{Martins1987, Klarbring1988,
  Ionescu2002}). Here, in the general case, we establish existence and
uniqueness by introducing viscosity, expressed as a Kelvin-Voigt
relaxation, with a small coefficient. For non-opening ruptures
following Lipschitz continuous rate- and state-dependent friction
laws, the further necessary conditions are natural.

At the same time, in recent years, vatious numerical algorithms have
been developed for computing solutions for rate- and state-dependent
friction laws coupled to the elastic-gravitational system of
equations, based on the above mentioned simplifications, seemingly
producing physically reasonable results \cite{Geubelle1995,
  Dalguer2007, Benjemaa2009, Pelties2012, Kozdon2013, Zhang2014,
  OReilly2015, Duru2016}. Some numerical studies do point out that
problems (like shock waves) can occur for long-time simulations, and
that introducing artificial viscosity is a natural way to obtain a
stable solution (\eg \cite{Day2005, Kaneko2008a, Aagaard2013}).
However, a mathematical framework addressing well-posedness on any
finite time interval through viscous regularization while avoiding
possibly unphysical simplifications to study coupled rupture dynamics
and seismic wave generation accounting for self gravitation, has
indeed been lacking so far and is presented here. The main technique
is iterative coupling the convergence of which we establish in concert
with the occurrence of two time scales. We suppress the uniform
rotation in our analysis, but including this is a simple task.

The friction law appears on some of the interior boundaries identified
as faults, and involves a nonlinear algebraic relation with the
evolution of a state variable that is represented by a time-dependent
nonlinear ordinary differential equation (ODE). Our approach is based
on considering an iterative coupling scheme derived from initially
decoupling the elastic-gravitational system of equations from the
state ODE supplemented with nonlinear frictional constraints. By
considering the equations satisfied by the difference of two
successive iterates, we obtain a contraction in natural norms for
these. The fixed point obtained as a result of this contraction solves
the original system of coupled equations. The artificial viscosity,
introduced in the elastic-gravitational system of equations, is used
in an essential manner to obtain the required estimates. Since the
proposed iterative scheme decouples two physically distinct problems,
a multi-rate scheme \cite{Girault2016} is a natural outcome. The
natural choice of numerical method is the discontinuous Galerkin (DG)
one, see earlier works by de la Puente \etal \cite{delaPuente2009},
Tago \etal \cite{Tago2012} and Pelties \etal \cite{Pelties2012}. In a
companion paper \cite{Ye2018ruptnum}, we develop a novel algorithm for
the multi-rate iterative coupling scheme proposed, here, using a nodal
DG method with penalty numerical flux enabling the general simulation
and studying of earthquakes.

The outline of this paper is as follows. In
Section~\ref{sec:assumptions}, we give the strong formulation for
particle motion, including self gravitation, and boundary conditions
expressing the coupling with a friction law, and the corresponding
weak formulation with necessary assumptions including the regularity
of model geometry and model parameters. The empirical assumptions of
friction laws are also discussed. In Section~\ref{sec:variational
  form}, we define the appropriate energy spaces, and then introduce
the variational form. In Section~\ref{sec:coupling}, we propose an
iterative coupling scheme and present a proof of contraction, with a
condition on the artificial viscosity. Based on the contracting
iterative scheme, we proof the existence of solutions for the coupled
problem in Section~\ref{sec:exist}. We discuss a backward Euler time
discretization in Section~\ref{sec:disc t}. The proof of contraction
is under certain conditions, which we precise in the theorems, for the
time step and the choice of viscosity coefficient.  Theorem~\ref{th1}
shows the contraction of iterative scheme in the continuous time case,
and Theorem~\ref{th2} in the discrete time case, with conditions on
the viscosity coefficient depending on model geometry and elastic
parameters, as we highlight in Remark \ref{rmk th2}. The contraction
result of Theorem~\ref{th1} allows us to deduce the existence of a
mixed solution to the coupled system as shown in Theorem
\ref{th:exist}.

\section{Mathematical model and assumptions}
\label{sec:assumptions}

We consider the problem in a bounded subdomain
$\overline\Domn\in\mathbb{R}^3$ that stands for the interior of the
solid earth (ignoring the fluid oceans and outer core), with a
continuum of linear elastic material that follows Hooke's law, except
at the rupture surface. We further assume that $\Domn$ is a Lipschitz
composite domain, which is defined as a disjoint union of open
subsets, $\Domn=\bigcup_{k=1}^{k_0} \Domn_k$, with interior boundaries
(supplemented with slip and non-slip conditions) given by
\[\Surf=\bigcup_{1\leq k<k'\leq k_0} 
    \partial\Domn_k\cap\partial\Domn_{k'} 
\setminus \partial\Domn,\]
which are two-dimensional Lipschitz continuous surfaces. 
We denote the rupture surface by $\Surf_\frc$, 
which is an open subset of $\Surf$.  
We have $\overline\Domn=\Domn\cup\Surf\cup\partial\Domn$.
The boundary of the interior surface, $\partial\Surf$, 
is a finite union of curves of measure 0 lying on the exterior boundary
$\partial \Domn$, 
where a traction free condition (\ref{eq:freesurf bc}) is applied.
We let $\ven\colon \partial\Domn_k \to \mathbb{R}^3$ be
the unit normal vector of interior and exterior boundaries,
defined almost everywhere on $\Surf \cup \partial \Domn$. 
It satisfies $\ven\in L^\infty(\Surf\cup\partial\Domn)^3$,
and labels the two sides of $\Surf$ by ``$\mm$'' and ``$\pp$''.
The jump operator $\smash{\jmp{\,\centerdot\,}}$ can be defined
for any bounded Lipschitz continuous function $f$ as
\begin{equation}
    \jmp{f}:=f^\pp-f^\mm=
    f^{\overline{\Domn_k}}(x)-f^{\overline{\Domn_{k'}}}(x),\quad
    \mbox{ for }x\in\partial\Domn_k\cap\partial\Domn_{k'},
    \label{eq:def pm}
\end{equation}
where $\Domn_{k}$ corresponds to the region of the ``$\pp$'' side and
$\Domn_{k'}$ to the region on the ``$\mm$'' side. 

\subsection{The basic equations in the strong form}
\label{sec:ass1}

We follow Brazda \etal \cite{Brazda2017} in introducing the equation
of motion in a prestressed earth while ignoring the rotation of the
earth. The gravitational potential $\phi^0$ satisfies Poisson's
equation
\begin{equation}
    \Delta\phi^0=4\pi \Grav\rho^0,
    \label{eq:grav pot}
\end{equation}
with $\rho^0$ the initial density distribution of the earth,
and $\Grav$ Newton's universal constant of gravitation.
The equilibrium condition for the initial steady state is 
\begin{equation}
    \rho^0\nabla\phi^0=\nabla\cdot\tsT^0,
    \label{eq:ini stress}
\end{equation}
where $\tsT^0$ is the tensor representing the static prestress. We
define $\veg^0:=\nabla\phi^0$, and the equation of motion is written
following \cite[(5.43)]{Brazda2017} as
\begin{equation}
    \rho^0\ddtt\veu+\rho^0\nabla \opS(\veu)
    + \rho^0\veu\cdot(\nabla\veg^0)
    -\nabla\cdot(\tsLaINI:\nabla\veu)=0
    \quad \mbox{in }\Domn\setminus\Surf_\frc
    \label{eq:strong form}
\end{equation}
with the initial conditions given as
$$\veu|_{t=0}=0 ,\ \ddt\veu|_{t=0}=0 .$$ 
The mass redistribution potential $ \opS(\veu) $ is 
associated with particle displacement $\veu$ by
\begin{equation}
    \Delta \opS(\veu)=-4\pi \Grav\nabla\cdot(\rho^0\veu),
    \label{eq:grav pot ptb}
\end{equation}
and the prestressed elasticity tensor is a linear map
$\tsLaINI\colon\mathbb{R}^{3\times3}\to\mathbb{R}^{3\times3}$ such
that $(\tsLaINI:\nabla\veu)$ represents the first Piola-Kirchhoff
stress perturbation. The prestressed elasticity tensor is related to
the \textit{in situ} isentropic elastic tensor $\tsC$ by
\[
    \itLaINI_{ijkl}=C_{ijkl}+\tfrac12\big(
     (T_0)_{ij}\delta_{kl}
    +(T_0)_{kl}\delta_{ij}
    +(T_0)_{ik}\delta_{jl}
    -(T_0)_{il}\delta_{jk}
    -(T_0)_{jk}\delta_{il}
    -(T_0)_{jl}\delta_{ik}
    \big).
\]
The non-slipping inner interfaces yield the conventional continuous
boundary conditions,
\begin{equation}
    \jmp{\veu}=0\,,\quad
    \jmp{\ven\cdot\big(\tsLaINI:\nabla\veu\big)}=0
    \quad \mbox{on } \Surf\setminus \fSurf,
    \label{eq:continuous bc}
\end{equation}
and the external boundary yields the traction free condition,
\begin{equation}
    \ven\cdot\big(\tsLaINI:\nabla\veu\big)^\mm=0
    \quad \mbox{on } \partial\Domn.
    \label{eq:freesurf bc}
\end{equation}
We denote by $\tsT_\delta(t,\vex)$ the perturbation of the stress
tensor away from the prestress $\tsT^0$ and subtracting the stress
variation generated by elastic motion.  In other words, the total
stress tensor can be expressed as $\tsT=\tsT^0+\tsT_\delta+(\tsLaINI :
\nabla\veu)$. Several models of $\tsT_\delta$ are available as
approximation of particular physical problems.  One popular model of
hydraulic fracturing considers poroelastic coupling of stress and pore
pressure due to injection of fluid mass \cite{Shapiro2009,
  Segall2015}, while another widely used model is based on thermal
pressurization \cite{Andrews2002, Rice2006, Schmitt2015}, where the
heat is generated by friction resistance to slow sliding and changes
the pore pressure of a fluid-saturated porous medium.  In both
scenarios, the governing equations are diffusive.  Therefore, we
safely assume that $T_\delta(t,\vex)$ is in $W^{1,2}([0,T], \spH)$,
with $\spH$ defined in Section \ref{sec:spaces}.

On the rupture surface
$\Surf_\frc$, the dynamic slip boundary condition (\eg
\cite[(4.57)]{Brazda2017}) and the force equilibrium are satisfied,
which give
\begin{equation}
    \left\{
        \arraycolsep=1.4pt\def\arraystretch{1.7}
        \begin{array}{rl}
            \jmp{\ven\cdot\veu} = & 0 ,
            \\
            \jmp{\vetau_1(\veu)+\vetau_2(\veu)} = & 0 ,
            \\
            \vetau_\frc-(\ven\cdot\big(\tsT^0+\tsTd\big)
            +\vetau_1(\veu)+\vetau_2(\veu))_\TT = & 0 ,
        \end{array}
        \right.
        \quad \mbox{on } \fSurf,
    \label{eq:fric bc}
\end{equation}
with
\begin{equation}
    \left\{
        \arraycolsep=1.4pt\def\arraystretch{1.7}
        \begin{array}{rl}
            \vetau_1(\veu)=&\ven\cdot\big(\tsLaINI:\nabla\veu\big) ,
            \\
            \vetau_2(\veu)=&
            -\nablaS\cdot \big(\veu\big(\ven\cdot\tsT^0\big)\big) ,
        \end{array}
        \right.
    \label{eq:fric bc var}
\end{equation}
both of which are linear functions depending on $\veu$, and the
surface divergence is defined by $\nablaS\cdot\vef
=\nabla\cdot\vef-(\nabla\vef\cdot\ven)\cdot\ven$.  We denote by
  $\sigma$ a scalar: its absolute value stands for the magnitude of
  normal stress, and it takes a positive/negative sign when the normal
  stress is compressive/expansive. Here, we view $\sigma :
  \mathbb{R}^3\to\mathbb{R}$ as a linear map, which maps particle
  displacement $\veu$ to the normal compression magnitude on the
  rupture surface.  That is, $\sigma$ is given by
\begin{equation}
    \sigma(\veu)= -\ven\cdot \big(\ven\cdot(\tsT^0+\tsTd)
    +\vetau_1(\veu)+\vetau_2(\veu)\big).
    \label{eq:def sigma}
\end{equation}
We also define the mean value of $\sigma(\veu)$ across $\Surf_\frc$ by
\begin{equation}
    \bar\sigma(\veu):=
    \tfrac12\big(\sigma(\veu^\pp)+\sigma(\veu^\mm)\big),
    \label{eq:def bsigma}
\end{equation}
which will be used in the construction of the variational form in
Problem~\ref{pb1}.

We denote by $\ves:=[\ddt\veu_\TT]_-^+$ the tangential jump 
of particle velocity across the rupture surface, 
and by $s:=|\ves|,\,\tau_\frc:=|\vetau_\frc|$
the slip-rate and the friction force magnitude, respectively.
The direction of friction force is opposite to the slip velocity, 
following (\eg \cite[eq.\ (4)]{Day2005})
\begin{equation}
    \tau_\frc\ves - s\vetau_\frc = 0.
\end{equation}
The nonlinear relation between $s$ and $\tau_\frc$ is governed by a
rate and state friction law, which we will discuss in
Section \ref{sec:RS law}.

We mention an equivalent description of the wave motion as an
alternative for the above equations (\ref{eq:strong form}),
(\ref{eq:continuous bc}), (\ref{eq:fric bc}) and (\ref{eq:fric bc
  var}). Within this description, the incremental Lagrangian stress
tensor takes the place of the incremental Piola-Kirchhoff stress
tensor, and the equation of motion attains the form (\eg
\cite[(5.52)]{Brazda2017})
\begin{equation}
    \rho^0\ddtt\veu+\rho^0\nabla \opS(\veu)
    - (\nabla\cdot(\rho^0\veu))\veg^0
    + \nabla\cdot(\veu\cdot\nabla\tsT^0)
    -\nabla\cdot(\tsGaINI:\nabla\veu)=0
    \quad \mbox{in }\Domn\setminus\Surf_\frc,
    \label{eq:strong form L}
\end{equation}
where $\tsGaINI\colon\mathbb{R}^{3\times3}\to\mathbb{R}^{3\times3}$ 
is a linear map such that $(\tsGaINI:\nabla\veu)$ represents the
first-order Lagrangian stress perturbation,
which satisfies the same boundary condition as 
(\ref{eq:continuous bc}) and (\ref{eq:fric bc}),
with $\vetau_1$ and $\vetau_2$ replaced by 
$\tilde\vetau_1$ and $\tilde\vetau_2$, given by
\begin{equation}
    \left\{
        \arraycolsep=1.4pt\def\arraystretch{1.7}
        \begin{array}{rl}
            \tilde\vetau_1(\veu)=&\ven\cdot(\tsGaINI:\nabla\veu) ,
            \\
            \tilde\vetau_2(\veu)=&
            -\ven\cdot(\veu\cdot\nablaS\tsT^0)
            -\tsT^0\cdot\nablaS(\ven\cdot\veu) .
        \end{array}
        \right.
    \label{eq:fric bc var L}
\end{equation}
The surface gradient is defined by $\nablaS \vef =\nabla \vef-(\nabla
\vef\cdot\ven)\ven$.  We can apply the same coupling scheme to
(\ref{eq:strong form L}) and (\ref{eq:fric bc var L}) and obtain
similar well-posedness results that will be developed in
Sections~\ref{sec:coupling}--\ref{sec:disc t}.

\subsection{Rate- and state-dependent friction laws}
\label{sec:RS law}

Here, we review the general assumptions for composing a rate- and
state-dependent friction law, which will be essential in the proof of
well-posedness of the coupling problem. A detailed discussion and
analysis can be found in Rice \etal \cite{Rice2001}. Upon introducing
a state variable $\psi$ that measures the average contact maturity,
the nonlinear relation for the magnitude of friction force can be
written in the general form of a scalar function
\begin{equation}
    \tau_\frc=\opF\big(\sigma,s,\psi\big).
    \label{eq:friction law}
\end{equation}
The state variable evolves in time according to the ordinary
differential relation,
\begin{equation}
    \dot \psi + 
    \opG\big(\sigma,\ddt\sigma,s,\psi\big)=0 .
    \label{eq:state ODE}
\end{equation}
A steady state can be obtained for each pair of $(\sigma,s)$ by taking
$\dot{s}=0$ and $\dot{\sigma}=0$, with a corresponding state-variable
value $\psi_\std(\sigma,s)$ satisfying
\begin{equation}
    \begin{split}
    \opG\big(\sigma,0,s,\psi_\std(\sigma,s)\big)=0,
    \end{split}
    \label{eq:steady stat}
\end{equation}
and with the corresponding friction force denoted by
\begin{equation}
    \begin{split}
    \tau_\std(\sigma,s):=\opF\big(\sigma,s,\psi_\std(\sigma,s)\big).
    \end{split}
    \label{eq:steady force}
\end{equation}
For the dynamic rupture problem considered here, we assume that the
rupture remains compressive, in other words, $\sigma$ stays positive
and the friction force $\vetau_\frc$ is non-vanishing if the slip rate
$s$ is non-zero. This assumption puts constraints on the ruptures,
physically meaning that the block mass across the fault should stay in
contact without any ``opening'' portion existing throughout time. This
assumption is naturally implied in typical rock physics experiments,
and applies to most research of earthquake processes. Furthermore, we
invoke

\begin{assumption}
    The nonlinear functions $\opF$ and $\opG$ in \eqref{eq:friction law}
    and \eqref{eq:state ODE} are uniformly Lipschitz continous
    in all the variables.
    \label{ass:Lipschitz op}
\end{assumption}

\begin{assumption}[Amontons-Coulomb law]
    The magnitude of instantaneous friction force is
    proportional to the compressive normal stress magnitude
    in the way that (\cf \cite[eq.\ (4a)]{Ruina1983})
    \begin{equation}
        \opF\big(\sigma,s,\psi\big)=\sigma \,f(s,\psi).
        \label{eq:Coulomb law}
    \end{equation}
    \label{ass:Coulomb law}
\end{assumption}

In the above, $f(s,\psi)$ is usually called the friction
coefficient. A physically meaningful friction coefficient is positive
and bounded, which indicates that $\tau_\frc$ always depends
positively on the magnitude of the compressive normal stress. Based on
experimental observations, it has also been recognized that the
instantaneous friction force depends positively on the slip rate $s$
as well as on the state variable $\psi$.  Together with Assumption
\ref{ass:Lipschitz op}, these empirical laws lead to the following

\begin{assumption}
    There exist positive constants 
    $C_{\opF,s}$, 
    $C_{\opF,s}^\star$, $C_{\opF,\sigma}$,
    $C_{\opF,\sigma}^\star$, $C_{\opF,\psi}$ and $C_{\opF,\psi}^\star$
    such that the nonlinear function $\opF$ in \eqref{eq:friction law}
    satisfies
\begin{equation}
    \begin{aligned}
    &    
    C_{\opF,s}^\star \geq 
    \pdd \opF s (\sigma,s,\psi)
    \geq C_{\opF,s} > 0
    ,\quad
    C_{\opF,\sigma}^\star \geq
    \pdd \opF\sigma(\sigma,s,\psi) 
    \geq C_{\opF,\sigma} >0
    ,
    \\&
    C_{\opF,\psi}^\star \geq
    \pdd \opF\psi(\sigma,s,\psi) 
    \geq C_{\opF,\psi} > 0
    \quad
    \mbox{for all } \sigma,s,\psi \in \mathbb{R}^+.
    \end{aligned}
    \label{eq:bnd F}
\end{equation}
\end{assumption}

The general features of the function $\opG$ are still under debate.
Studies by Linker and Dieterich \cite{Linker1992}, Prakash
\cite{Prakash1998}, Richardson and Marone \cite{Richardson1999},
Bureau \etal \cite{Bureau2000}, and many others show that the effects
of variable compressive normal stress upon friction state can take
various forms. Instead, we use the laws of Dieterich-Ruina
\cite[p.\ 1875]{Rice2001}, which ignore the dependency on variational
normal stress of the nonlinear state ODE (\ref{eq:state ODE}).
In other words, we replace \eqref{eq:state ODE} by a simplified form
\begin{equation}
    \ddt\psi+\opG(s,\psi)=0.
    \label{eq:state ODE1}
\end{equation}
Meanwhile, empirical results from laboratory experiments suggest that
there is a characteristic length for the steady-sliding rupture
evolving into the next steady state after a sudden change of slip
rate, regardless of the value of slip rate.  Elaboration on this
observation follows linearizing (\ref{eq:state ODE1}) as a
perturbation of steady state, which yields (\cf
\cite[eq.\ (7)]{Ruina1983})
\begin{equation}
    \DDt \psi=-\frac{\partial \opG}{\partial \psi} 
    \left( \psi -\psi_\std \right),
    \label{}
\end{equation}
with a solution (\cf \cite[eq.\ (8)]{Ruina1983})
\begin{equation}
    \psi(s,L/s)=\psi_\std(s)+\big(\psi(s,0)-\psi_\std(s)\big)
    \exp\left( -\frac{L}{s} 
    \frac{\partial \opG}{\partial \psi} \right),
    \label{}
\end{equation}
in which the time is replaced by $L/s$, where $L$ is the slip
distance.  The characteristic length is defined as $L_c:=s/(\partial
\opG/\partial \psi)$, physically meaning that after slipping for a
distance $L_c$ under fixed compressive normal stress and slip rate,
the friction coefficient evolves towards the steady state by a
definite ratio $1/e$. The empirical law above indicates that $L_c$ is
independent of $s$, and a linear slip-dependent friction law can be
regarded as a trivial interpretation by letting $\opG=s/L_c$.  The
non-negative nature of $L_c$ and $s$ implies the following

\begin{assumption}
    There exist non-negative constants 
    $C_{\opG,\psi}$,
    $C_{\opG,\psi}^\star$ and 
    $C_{\opG,s}^\star$ such that the nonlinear function
    $\opG$ in \eqref{eq:state ODE1} satisfies
    \begin{equation}
    0\leq C_{\opG,\psi}\leq
    \pdd\opG\psi(s,\psi)
    \leq C_{\opG,\psi}^\star
    ,\quad
    \left|\pdd \opG s (s,\psi)\right| \leq C_{\opG,s}^\star 
    \quad
    \mbox{for all } s,\psi \in \mathbb{R}^+.
        \label{eq:bnd G}
    \end{equation}
\end{assumption}

\subsection{Assumptions on material parameters}
\label{sec:assumption}

We give assumptions on the regularity of parameters following
\cite{deHoop2015}. The reference density, $\rho^0$, is contained in
$L^\infty(\overline\Domn)\cap W^{1,\infty}(\Domn)$, where
$W^{1,\infty}$ is the space of $C^0$ functions whose weak gradient is
in $L^\infty$, and
\[
    \left\{
        \begin{array}{ll}
            C_{\rho^0}^\star \geq \rho^0(\vex)\geq C_{\rho^0} >0 ,
            &\quad \vex\in\overline\Domn\\
            \rho^0(\vex) \equiv 0,&\quad \vex\in\overline\Domn^c ;
        \end{array}
    \right.
\]
thus $\phi^0\in W^{2,2}(\mathbb{R}^3)$ by elliptic regularity. 
The prestress tensor
$\tsT^0 $ 
governed by (\ref{eq:ini stress}) satisfies the symmetries
\[
   (T_0)_{ij} = (T_0)_{ji},\quad i,j\in\{1,2,3\},
\]
and the continuity on interfaces
\[
   \jmp{\ven\cdot\tsT^0}=0.
\]
The stiffness tensor $C_{ijkl}\in L^\infty(\overline\Domn)
^{3\times3\times3\times3}$ satisfies the symmetries
\[
    C_{ijkl}=C_{klij}=C_{jikl}=C_{ijlk},\quad 
    i,j,k,l\in\{1,2,3\}.
\]
It automatically follows that $\tsLaINI\in L^\infty(\overline\Domn)
^{3\times3\times3\times3}$, which also satisfies the symmetry relation
\[
    \itLaINI_{ijkl}=\itLaINI_{klij},\quad
    i,j,k,l\in\{1,2,3\}.
\]
Moreover, we borrow the following assumptions from 
{\it de Hoop, et al} \cite{deHoop2015}
(the assumptions of Theorem 2)

\begin{assumption}
    We have
    \begin{enumerate}
        \item $(T_0)_{ij}\in L^\infty(\Domn)$, 
            with $\mathrm{Tr}(\tsT^0):=\sum_{i\in\{1,2,3\}}(T_0)_{ii}$
            bounded away from 0;
        \item $\veg^0 \in L^\infty(\Domn)$ with $\norm{\veg^0}$ 
            bounded away from 0;
        \item there exists $\mathfrak{c}>0 $ such that for any 2-tensors
            $\eta_{ij}$, 
            \[
            \mathfrak{c}|\eta_{ij}+\eta_{ji}|^2\leq
        \itLaINI_{ijkl}
        \eta_{kl}\eta_{ij}.
            \]
    \end{enumerate}
    \label{ass:T0}
\end{assumption}

\begin{remark}
We ignore the liquid regions including the outer core and ocean layer.
The analysis of a self-gravitating planet with fluid regions can be
found in {\it de Hoop, et al} \cite{deHoop2015}.  Including fluid
regions does not harm the well-posedness of the coupled problem as
long as the intersections of $\Surf_\frc$ and fluid-solid interfaces
are finite curves with zero measure, which do not appear in the
analysis.
\end{remark}

\section{The variational form}
\label{sec:variational form}

We bring the overall problem in variational form with coupling with a
nonlinear algebraic relation and time evolution of state on the
interior slipping boundary or rupture plane. In this section, we
present the procedure and introduce the relevant Sobolev spaces.

\subsection{Energy spaces and trace theorem}
\label{sec:spaces}

In the Lipschitz composite domain $\Domn\subseteq \mathbb{R}^3$,
the space of square integrable functions is defined as
\[
    L^2(\Domn)=\bigg\{ v\,\bigg| \, 
        \sum_{k=1}^{k_0} \norm{v}_{L^2(\Domn_k)}^2<\infty \bigg\}.
\]
We define the Sobolev space $\spH$ as
\[
    \begin{split}
    \spH=&\bigg\{\vev\in L^2(\Domn)^3
        \,\bigg|\,
        \sum_{k=1}^{k_0} \norm{\nabla \vev}_{L^2(\Domn_k)}^2<\infty
    \bigg\} ,
    \end{split}
\]
with the norm
\[
    \norm{\vev}_{\spH}:=\bigg(
    \sum_{k=1}^{k_0} \norm{\vev}_{H^1(\Domn_k)}^2\bigg)^{1/2}
    ;
\]
we denote its dual space with regard to $L^2(\Domn)$ by $\spH'$. With
Assumption \ref{ass:T0} (3), the 4-tensor $\tsLaINI$ is convex. We
denote by $\norm{\centerdot}_{L^2(\Domn;\rho^0)}$ the weighted norm
\begin{equation}
    \begin{split}
        \norm{\veu}_{L^2(\Domn;\rho^0)}:=&
    \sum_{k=1}^{k_0}
    \int_{\Domn_k} 
    \rho^0|\veu|^2\dd\Domn.
    \end{split}
    \label{eq:weighted norm}
\end{equation}
Clearly, under the assumptions introduced in
Section~\ref{sec:assumption}, $\norm{\centerdot}_{L^2(\Domn)}^2$ and
$\norm{\centerdot}_{L^2(\Domn;\rho^0)}^2$ are equivalent,
\begin{equation*}
    \begin{array}{rclcl}
    C_{\rho^0}\norm{\vev}_{L^2(\Domn)}^2 
    & \leq &
    \norm{\vev}_{L^2(\Domn;\rho^0)}^2 
    & \leq &
    C_{\rho^0}^\star\norm{\vev}_{L^2(\Domn)}^2 ,
\quad \forall \vev\in L^2(\Domn).
    \end{array}
\end{equation*}
We denote by $\langle\,,\,\rangle_{\Surf_\frc}$ the duality pairing
between $H^{-1/2}(\Surf_\frc)$ and $H^{1/2}(\Surf_\frc)$, and
by $\langle\,,\,\rangle_{\Domn}$ the duality pairing between $\spH$
and its dual $\spH'$. We study the weak solution of the coupling
problem in the space $V_1 \times V_2$,
\begin{equation}
    \begin{split}
        V_1:=&\left\{
            \veu \in 
            L^{\infty}\big([0,T];\spH\big)
            \,\left\vert\,
            \begin{array}{l}
            \ddt\veu \in 
            L^\infty\big([0,T];L^2(\Domn)\big)\cap
            L^2\big([0,T];\spH\big),
            \\[2mm]
            \ddtt\veu \in L^2\big([0,T];\spH'\big),
            \\[2mm]
            \jmp{\ven\cdot\veu}=0\mbox{ on } \Surf_\frc 
            \end{array}
            \right.
        \right\},
        \\
        V_2:=&\left\{
            \psi \in L^{\infty}\big([0,T];L^2(\Surf_\frc)\big)
            \,\Big\vert\,
            \ddt\psi \in L^2\big([0,T];L^2(\Surf_\frc)\big)
        \right\}.
    \end{split}
    \label{eq:space of sol}
\end{equation}

We revisit the general Sobolev trace theorem (\eg \cite[Theorem
  1.3.1]{Quarteroni2008}) and rewrite it for interior boundaries. The
quantities $\vev^\pm$ related to any $\vev\in \spH$ are defined in
(\ref{eq:def pm}).

\begin{lemma}
    Let $\Surf_{\frc_{k,k'}} = \partial\Domn_k\cap \partial\Domn_{k'} 
    \setminus \partial \Domn$ be a Lipschitz continuous interior 
    boundary for two adjacent subdomains $\Domn_k$ and $\Domn_{k'}$.
    \begin{enumerate}[label=(\alph*)]
    \item There exist two unique linear continuous maps (trace
      operators) $T_{\frc_{k,k'}^\pp}: H^1(\Domn_k)^3 \linebreak \to
      H^{1/2}(\Surf_{\frc_{k,k'}})^3$ and $T_{\frc_{k,k'}^\mm}:
      H^1(\Domn_{k'})^3\to H^{1/2}(\Surf_{\frc_{k,k'}})^3$, such that
      $T_{\frc_{k,k'}^\pp}(\vev)=\vev^\pp|_{\Surf_{\frc_{k,k'}}}$ and
      $T_{\frc_{k,k'}^\mm}(\vev)=\vev^\mm|_{\Surf_{\frc_{k,k'}}}$ for
      each $\vev\in \spH$.
    \item There exist two linear continuous maps (extension operators)
      $R_{\frc_{k,k'}^\pp}: H^{1/2}(\Surf_{\frc_{k,k'}})^3 \linebreak
      \to H^1(\Domn_k)^3$ and $R_{\frc_{k,k'}^\mm}:
      H^{1/2}(\Surf_{\frc_{k,k'}})^3 \to H^1(\Domn_{k'})^3$, such that
      $ T_{\frc_{k,k'}^\pp}\circ R_{\frc_{k,k'}^\pp}(\vev)=
      T_{\frc_{k,k'}^\mm}\circ R_{\frc_{k,k'}^\mm}(\vev)=\vev$, for
      each $\vev\in H^{1/2}(\Surf_{\frc_{k,k'}})^3$.
    \end{enumerate}
    \label{th:trace thm}
\end{lemma}

This lemma implies the existence of constants $C_{\frc_{k,k'}}^\pm>0$
such that
\begin{equation}
    \begin{split}
        &
    \norm{T_{\frc_{k,k'}^\pp}(\vev)}_{L^2(\Surf_{\frc_{k,k'}})}^2\leq 
    C_{\frc_{k,k'}^\pp}\norm{\vev}_{H^1(\Domn_k)}^2 
    \mbox{ and }
    \\&
    \norm{T_{\frc_{k,k'}^\mm}(\vev)}_{L^2(\Surf_{\frc_{k,k'}})}^2\leq 
    C_{\frc_{k,k'}^\mm}\norm{\vev}_{H^1(\Domn_{k'})}^2,
    \quad\forall v\in \spH.
    \end{split}
    \label{eq:trace ineq}
\end{equation}
We denote by $T_{\frc}$ the direct union of all 
$T_{\frc_{k,k'}^\pm}$,
and $C_\frc=\max_{(k,k';\pm)} C_{\frc_{k,k'}^\pm}$.
We can then define the tangential
jump operator $T_\frcpm$ for interior boundaries that generates
$\ves=T_\frcpm(\ddt\veu)$ and yields the following lemma, which can be
obtained directly from Lemma \ref{th:trace thm}. 

\begin{lemma}
    Let $\Domn$ be a Lipschitz composite domain 
    and $\Surf_\frc$ be a subset of its 
    Lipschitz interior boundaries.
    \begin{enumerate}[label=(\alph*)]
        \item There exists a unique linear continuous map 
            $T_\frcpm: \spH\to 
                H^{1/2}(\Surf_\frc)^3$ 
            such that $T_{\frcpm}(\vev)=\jmp{\vev_\TT}$,
            for each $\vev\in \spH$.
        \item There exists a linear continuous map 
            $R_{\frcpm}: H^{1/2}(\Surf_{\frc})^3 
            \to \spH$ such that
            $T_{\frcpm}\circ R_{\frcpm}(\vev)=\vev$ for each 
            $\vev\in H^{1/2}(\Surf_{\frc})^3$.
        \item There exists a constant $C_{\frcpm}>0$ such that
        \begin{equation}
            \norm{\jmp{\vev_\TT}}_{L^2(\Surf_{\frc})}^2=
            \norm{T_{\frcpm}(\vev)}_{L^2(\Surf_{\frc})}^2\leq 
            C_{\frcpm}\norm{\vev}_{\spH}^2,
            \quad\forall \vev\in \spH.
            \label{eq:trace ineq jmp}
        \end{equation}
    \end{enumerate}
    \label{th:trace thm jmp}
\end{lemma}

We introduce the (bounded linear) Dirichlet-to-Neumann maps
\cite{Salo2008,Beretta2014,Beretta2017} associated with the
elastic-gravitational system of equations (\ref{eq:strong form}),
\begin{equation}
    \begin{split}
        &
    \itLa_{\tsLaINI,\rho^0,\veg^0}:
    H^{1/2}(\fSurf)^3 \ni T_\frc (\veu) \to 
    \big(\ven\cdot(\tsLaINI:\nabla\veu)\big)\big|_{\fSurf}
    \ \in H^{-1/2}(\fSurf)^3,
    \\&
    \itLa'_{\tsLaINI,\rho^0,\veg^0}:
    H^{1/2}(\fSurf)^3 \ni T_\frc (\veu) \to 
    \big(\nablaS\cdot(\veu(\ven\cdot\tsT^0))\big)\big|_{\fSurf}
    \in H^{-1/2}(\fSurf)^3.
    \end{split}
    \label{eq:def opLa}
\end{equation}
Clearly, 
\begin{equation}
    \begin{split}
        &
    \norm{\vetau_1}_{H^{-1/2}(\fSurf)}^2=
    \norm{\itLa_{\tsLaINI,\rho^0,\veg^0}\circ T_\frc\,\big(\veu\big)}
    _{H^{-1/2}(\fSurf)}^2
    \leq C_{\itLa} \norm{\veu}_{H^{1/2}(\fSurf)}^2,
    \\ &
    \norm{\vetau_2}_{H^{-1/2}(\fSurf)}^2=
    \norm{\itLa'_{\tsLaINI,\rho^0,\veg^0}\circ T_\frc\,\big(\veu\big)}
    _{H^{-1/2}(\fSurf)}^2
    \leq C_{\itLa'} \norm{\veu}_{H^{1/2}(\fSurf)}^2.
    \end{split}
    \label{eq:DtN bnd}
\end{equation}
We obtain the following lemma which will be used in the proofs of 
Theorems \ref{th1} and \ref{th2}.

\begin{lemma}
\label{lm}
Let $T_\frc$, $T_\frcpm$, $\itLa_{\tsLaINI,\rho^0,\veg^0}$ 
and $\itLa'_{\tsLaINI,\rho^0,\veg^0}$
be defined as in Lemma \ref{th:trace thm}, Lemma \ref{th:trace thm jmp}, 
and equation \eqref{eq:def opLa},
then there exist constants $C_I,C_I'>0$ such that 
    \begin{equation}
        \begin{split}
        \big\langle
        \itLa_{\tsLaINI,\rho^0,\veg^0}\circ T_\frc \,(\veu)\,,\,
        T_\frcpm (\vev)\big\rangle_{\Surf_\frc} &\leq 
        C_I\norm{\veu}_{\spH}\norm{\vev}_{\spH},
\\
        \jmp{\big\langle
        \itLa'_{\tsLaINI,\rho^0,\veg^0}\circ T_\frc \,(\veu)\,,\,
        T_\frc (\vev)\big\rangle_{\Surf_\frc}} &\leq 
        C'_I\norm{\veu}_{\spH}\norm{\vev}_{\spH},\quad
        \forall \veu,\vev\in \spH\,.
        \end{split}
        \label{eq:space intp}
    \end{equation}
\end{lemma}

\begin{proof}
Based on the Cauchy-Schwartz inequality \cite{Salo2008},
\begin{equation}
    \big\langle
    \itLa_{\tsLaINI,\rho^0,\veg^0}\circ T_\frc \,(\veu)\,,\,
    T_\frcpm (\vev)\big\rangle_{\Surf_\frc}\leq 
    \norm{\itLa_{\tsLaINI,\rho^0,\veg^0}\circ T_\frc \,(\veu)}
    _{H^{-1/2}(\Surf_\frc)}
    \norm{T_\frcpm (\vev)}_{H^{1/2}(\Surf_\frc)}.
    \label{eq:th 4 pf}
\end{equation}
Using (\ref{eq:trace ineq}), (\ref{eq:trace ineq jmp}) and 
(\ref{eq:DtN bnd}) in (\ref{eq:th 4 pf}), we immediately obtain
\[
    \big\langle
    \itLa_{\tsLaINI,\rho^0,\veg^0}\circ T_\frc \,(\veu)\,,\,
    T_\frcpm (\vev)\big\rangle_{\Surf_\frc}\leq 
    (C_{\itLa}C_\frc C_\frcpm)
    \norm{\veu}_{\spH}\norm{\vev}_{\spH}\,.
\]
Thus $C_I=C_{\itLa}C_\frc C_\frcpm$. We can prove the second
inequality in (\ref{eq:space intp}) in the same manner.
\end{proof}

\subsection{A weak form of the system of equations and viscosity
            solutions}

We introduce the weak form on $\Domn$ while requiring the nonlinear
friction law to hold pointwise. We then follow the approach of
Martins and Oden \cite{Martins1987} and Ionescu \etal
\cite{Ionescu2003} to prove the well-posedness.

We introduce a convex and G\^ateaux differentiable approximation to
friction force $\vetau_\frc$ 
by defining the regularized slip rate as 
(\cf \cite[eq. (30)]{Ionescu2003})
\begin{equation}
    \sPsi(\vev)=\sqrt{|\vev|^2+\veps^2}-\veps
    \label{eq:psi-eps}
\end{equation}
with a small constant $\veps > 0$, whose gradient with regard
to the slip velocity is denoted by
\begin{equation}
    \bPsi(\vev)= \frac\vev{\sqrt{|\vev|^2+\veps^2}}.
    \label{eq:D-eps}
\end{equation}
Clearly, 
\begin{align}
    0\leq \sPsi(\vev) \leq |\vev|,&\quad
    \forall \vev\in \mathbb{R}^3 ,
    \label{eq:bnd Psi}
    \\
    |\bPsi(\vev)\cdot\vew| \leq |\bPsi(\vev)|\,|\vew|\leq |\vew|,&\quad
    \forall \vev,\vew \in \mathbb{R}^3 ,
    \label{eq:bnd dPsi}
    \\
    \big| \sPsi(\vev) - |\vev| \big|\leq \veps, &\quad
    \forall \vev\in \mathbb{R}^3 .
    \label{eq:bnd Psi err}
\end{align}
We then introduce the nonlinear map
$\scalF:H^{-1/2}(\Surf_\frc)\times L^2(\Surf_\frc)\times 
\spH \times \spH \to \mathbb{R}$ 
as a family of regularized friction functionals,
\[
    \begin{split}
        &
    \scalF(\sigma,\psi,\veu,\vev)=
    \int_{\Surf_\frc} \opF\big(\sigma,|T_\frcpm(\veu)|,\psi\big)\,
    \sPsi\big(T_\frcpm(\vev)\big)\dd\Surf,
    \\&\hspace{2cm}
    \text{for all}\quad\sigma\in H^{-1/2}(\Surf_\frc),\quad
    \psi \in L^2(\Surf_\frc),
    \quad \veu,\vev \in \spH\,.
    \end{split}
\]
We denote by $\bcalF:H^{-1/2}(\Surf_\frc) \times L^2(\Surf_\frc)
\times \spH \to H^{-1/2}(\Surf_\frc) $ the derivative of $\scalF$
with respect to the final variable such that
\[
    \big\langle\bcalF(\sigma,\psi,\vev),\vew \big\rangle_{\Surf_\frc}
    =\int_{\Surf_\frc} \opF\big(\sigma,|T_\frcpm(\vev)|,\psi\big)\,
    \bPsi\big(T_\frcpm(\vev)\big)\cdot T_\frcpm(\vew) \dd\Surf.
\]
In other words, $\bcalF(\sigma,\psi,\vev)$ represents the regularized
replacement of $\vetau_\frc$.

We write (\ref{eq:strong form})-(\ref{eq:fric bc var}) in the
following weak form, appended with an artificial (temporal) viscosity
term weighted by $\gamma > 0$ and a small and fixed regularization
coefficient $\veps$ upon the friction law.

\medskip\medskip

\begin{problem} \label{problem:1}
Let $\veps$ and $\gamma$ be fixed strictly positive constants,
    find $ (\veu,\psi) \in V_1\times V_2 $ such that
\begin{align}
        &
    \begin{aligned}
        &
    \big\langle \rho^0\ddtt\veu \,,\, \vew \big\rangle_{\Domn}
    +a_3(\veu \,,\, \vew)
    -\frac1{4\pi \Grav}
    \innp{ \nabla \opS(\veu) \,,\, \nabla \opS(\vew) }
    _{L^2(\mathbb{R}^3)}
    +\gamma \innp{\ddt\veu \,,\, \vew }_{\spH}
    \\ & \hspace{1cm}
    +\big\langle\bcalF\big(\bar\sigma(\veu),\psi,
        \ddt\veu\big)
    \,,\, \vew \big\rangle_{\Surf_\frc}
    -\jmp{\big\langle\vetau_2(\veu) \,,\, \vew \big\rangle
    _{\Surf_\frc} }
    \\ & \hspace{0.5cm}
    =
    \big\langle\ven\cdot(\tsT^0+\tsTd) \,,\, 
    T_\frcpm(\vew) \big\rangle_{\Surf_\frc} ,
    \end{aligned}
    \label{eq:weakform}
    \\[2mm] &
    \innp{ \ddt \psi,\varphi }_{L^2(\Surf_\frc)}
    +
    \innp{ \opG\big(|T_\frcpm(\ddt\veu)|,
        \psi\big),
    \varphi }_{L^2(\Surf_\frc)}
    =0 ,
    \label{eq:weakform ODE}
\end{align}
are satisfied almost everywhere in time,
with sesquilinear form $a_3(\veu,\vev)$ defined by
\begin{align}
    \begin{aligned}
        &
    a_3(\veu\,,\,\vev)=
    \int_\Domn (\tsLaINI:\nabla\veu):\nabla\vew \dd\Domn
    \\ & \hspace{1cm}
    -\int_\Domn \Sym\Big\{(\veg^0\cdot\veu)(\vev\cdot\nabla\rho^0)
        +\rho^0(\veg^0\cdot\veu)(\nabla\cdot\vev)
    +\rho^0\veg^0\cdot(\nabla\veu)\cdot\vev\Big\}\dd\Domn,
    \end{aligned}
    \label{eq:a2}
\end{align}
and the linear maps $\vetau_2$ and $\bar\sigma$ defined in
\eqref{eq:fric bc var} and \eqref{eq:def sigma} on $\Surf_\frc$ in the
sense of traces, and hold for all $(\vew,\varphi) \in V_1\times V_2 $.
\label{pb1}
\end{problem}

\medskip\medskip

\noindent
In the above, $\Sym$ is the symmetrization such that for any expression
$B(\veu,\vev)$, we have 
\[
    \Sym\{B(\veu,\vev)\}=\tfrac12\big(
    B(\veu,\vev)+B(\vev,\veu)\big).
\]
The proof of consistency between the strong form \eqref{eq:strong
  form} and the weak form \eqref{eq:weakform} can be found in {\it de
  Hoop, et al} \cite[Lemma 3]{deHoop2015}.  Moreover,
$a_3(\veu\,,\,\vev)$ is coercive

\begin{lemma}
    With the assumptions in Section~\ref{sec:assumption}, there exist 
    $C_{a_3},C_{a_3}'>0$ such that for all $\veu\in \spH$,
    \begin{equation}
        a_3(\veu\,,\,\veu)\geq C_{a_3} \norm{\veu}_{\spH}^2
        -C_{a_3}'\norm{\veu}_{L^2(\Domn)}^2.
        \label{eq:coercivity a3}
    \end{equation}
    \label{th:coercivity}
\end{lemma}

The proof of Lemma \ref{th:coercivity} is contained in the proof of
\cite[Theorem 2]{deHoop2015}. We further define
\begin{equation}
    a_3'(\veu\,,\,\vev):=a_3(\veu\,,\,\vev)
    -C_{a_3}'\big(\veu\,,\,\vev\big)_{L^2(\Domn)}.
    \label{eq:a3p}
\end{equation}
From Lemma \ref{th:coercivity} it is clear that also
\begin{equation}
a_3'(\veu\,,\,\veu)\geq C_{a_3} \norm{\veu}_{\spH}^2
\mbox{ holds for all } \veu\in \spH.
    \label{eq:coercivity a3p}
\end{equation}
Moreover, $a_3'$ is bounded \cite[Lemma 7]{deHoop2015})

\begin{lemma}
    Suppose the assumptions introduced in Section~\ref{sec:assumption} hold,
    then the sesquilinear form $a_3'$ is bounded, that is,
    \[
        a_3'(\veu,\vew)\leq C_{a_3}^* 
        \norm{\veu}_{\spH}
        \norm{\vew}_{\spH};\quad
        \forall \veu, \vew \in \spH\,,
    \]
    for some constant $C_{a_3}^* >0$, and is Hermitian.
    \label{lm:a3p bound}
\end{lemma}

The solution $(\veu,\psi)$ depends on the small parameter $\veps$
defined in \eqref{eq:psi-eps} as well as on the viscosity coefficient
$\gamma$. We suppress these dependencies in our notation in the
further analysis.

\begin{remark}
    In the formulation of Problem \ref{pb1}, the boundary conditions
    (\ref{eq:continuous bc}), (\ref{eq:freesurf bc}) and (\ref{eq:fric
      bc}) are enforced by surface integration. Since both $\Surf\cap
    \partial \Domn$ and $\fSurf\cap (\Surf\setminus\fSurf)$ are union
    of curves with zero Lebesgue measure, discontinuities that occur on these
    curves will not appear in the analysis.
\end{remark}

\section{Nonlinear coupling iterative scheme}
\label{sec:coupling}

Here, we present a robust convergent nonlinear coupling iterative
scheme. There are several considerations that underlie the
introduction of such a scheme. First, it simplifies the stability
analysis through studying the behaviors of each of the
subproblems. Secondly, it enables acceleration of solving the system
through introducing preconditioners for each of the
subproblems. Moreover, in the time discretization, it facilitates the
use of different time steps; this is critically important, since the
ruptures and wave propagation take place on significantly different
time scales. Thirdly, we obtain a proof of well-posedness by verifying
whether the iterative coupling is a contraction.

\subsection{The iterative scheme in weak form}

The iterative scheme is described by the following

\medskip\medskip

\begin{problem}
Let $\veps$ and $\gamma$ be fixed, strictly positive constants,
    and let $(\veu^{k-1},\psi^{k-1}) \in V_1\times V_2$
    be the solution generated from the previous iteration,
    find $(\veu^{k},\psi^{k})  \in V_1\times V_2$ 
    such that for all $(\vew,\varphi)\in V_1\times V_2$,
    almost everywhere in time, the following equations
    are satisfied,
    \begin{align}
        &
    \begin{aligned}
        &
    \big\langle \rho^0\ddtt\veu^k \,,\, \vew \big\rangle_{\Domn}
    +a_3\big(\veu^k\,,\,\vew\big)
    -\frac1{4\pi \Grav} 
    \innp{ \nabla \opS(\veu^{k-1}) \,,\, \nabla \opS(\vew) }
    _{L^2(\mathbb{R}^3)}
    +\gamma\innp{\ddt\veu^k , \vew }_{\spH}
    \\ & \hspace{1cm}
    +\big\langle\bcalF\big(\bar\sigma(\veu^{k-1}),\psi^{k-1},
    \ddt\veu^{k}\big)\,,\, \vew \big\rangle
    _{\Surf_\frc}
    -\jmp{\big\langle\vetau_2(\veu^{k-1}) \,,\, \vew 
    \big\rangle_{\Surf_\frc}}
    \\ & \hspace{0.5cm}
    = \big\langle\ven\cdot(\tsT^0+\tsTd) \,,\, T_\frcpm(\vew) 
    \big\rangle_{\Surf_\frc}
    \end{aligned}
        \label{eq:step 1}
        \\&
    \big( \ddt \psi^k,\varphi \big)_{L^2(\Surf_\frc)}
    +
    \big( \opG(|T_\frcpm(\ddt\veu^k)|,\psi^k),\varphi 
    \big)_{L^2(\Surf_\frc)}
    =0,
    \label{eq:step 2}
    \end{align}
    with the initial conditions, independent of $k$,
    \begin{equation}
        \veu^k\big|_{t=0}=0,\,\ddt\veu^k\big|_{t=0}=0 \mbox{ and }
        \psi^k\big|_{t=0}=\psi^0,
        \mbox{ with } \psi^0\in L^2(\Surf_\frc).
        \label{}
    \end{equation}
    \label{pb steps}
\end{problem}

\medskip\medskip

\noindent
In short, the updated variables from iteration $k-1$ are used in
computing the Neumann boundary condition for iteration $k$ until
convergence. In the next subsection, we give a convergence proof that
involves a bound on the viscosity coefficient, $\gamma$, in terms of
the material parameters and the trace constant.

\begin{remark}
    Following classical techniques (\eg Martins and Oden
    \cite{Martins1987}), the existence of a solution for the decoupled
    system is clear. In other words, for given
    $(\veu^{k-1},\psi^{k-1}) \linebreak \in V_1\times V_2$, the
    existence of a solution $\veu^k\in V_1$ to \eqref{eq:step 1}
    holds, and, meanwhile, given $\veu^k\in V_1$, there exists a
    solution $\psi^{k} \in V_2$ to \eqref{eq:step 2} due to the
    Lipschitz continuous right-hand-side of the ODE.
\end{remark}

\subsection{Convergence}
\label{sec:conv}

We show that the iterative coupling scheme described by Problem
\ref{pb steps} is linearly convergent, by a given convergence rate
$\lambda\in(0,1)$, within the space $V_1\times V_2$ for any finite
time interval $[0,T]$ under certain conditions.

\begin{theorem}
Let the maximal time $T$ and the coefficient $\gamma$ satisfy
\begin{equation}
    \begin{split}
    &
\frac1{\beta(T)}
\geq \max \left(
\frac{C_{\opF,\psi}^{\star\,2}}{\lambda C_\veps C_{\opF,s}}
+\left(\frac{C_{\opG,s}^{\star\,2}}{C_\veps C_{\opF,s}}
-2C_{\opG,\psi}\right)
\,,\,
\left(\frac{C_S C_{\rho^0}^{\star}}{4\pi \Grav }+C_{a_3}'
\right)\bigg/C_{\rho^0}
\,\right),
\\&
   \gamma \geq 2 \beta(T)
          (( (C_I + C'_I) C_{\opF,\sigma}^{\star})^2
     + C_I^{'\,2} ) / (\lambda (C_{a_3} - \beta(T) C_{a_3}'))
\end{split}
\label{eq:contr cond}
\end{equation}
with $\beta(T):\mathbb{R}^+\to\mathbb{R}^+$ a monotonically increasing
function of $T$, and $\lambda \in (0, 1)$ some constant.  Our
iterative coupling scheme described by Problem \ref{pb steps} is a
contraction in the sense that
\begin{equation}
\begin{aligned}
    &
  \kappa_1\norm{\ddt\veer_\veu^k}^{2}_{L^\infty([0,T];L^2(\Domn))}
+ \kappa_2\norm{\veer_\veu    ^k}^{2}_{L^\infty([0,T];\spH)}
+ \kappa_3\norm{\scer_\psi    ^k}^{2}_{L^\infty([0,T];L^2(\Surf_\frc))}
+ \kappa_4\norm{\ddt\veer_\veu^k}^{2}_{L^2([0,T];\spH)}
\\&\hspace{1cm}
\leq \lambda\Big(
  \kappa_1\norm{\ddt\veer_\veu^{k-1}}^{2}_{L^\infty([0,T];L^2(\Domn))}
+ \kappa_2\norm{\veer_\veu    ^{k-1}}^{2}_{L^\infty([0,T];\spH)}
+ \kappa_3\norm{\scer_\psi    ^{k-1}}^{2}_{L^\infty([0,T];L^2(\Surf_\frc))}
\Big),
\end{aligned}
\label{eq:contraction 1}
\end{equation}
where
$
\veer_\veu^k := \veu^k-\veu^{k-1}, \,
\scer_\psi^k := \psi^k-\psi^{k-1},
$
and
\begin{equation}
    \begin{aligned}
        &
    \kappa_1=C_{\rho^0}
    -\left(C_{a_3}'+\frac{C_S C_{\rho^0}^{\star}}{4\pi \Grav}\right)
    \beta(T)
    ,\quad
    \kappa_2=C_{a_3}-C_{a_3}'\beta(T),
    \\&
    \kappa_3=
    1-\left(\frac{C_{\opG,s}^{\star\,2}}{C_\veps C_{\opF,s}}
    -2C_{\opG,\psi}\right)\beta(T),
    \\&
    \mbox{and }
    \kappa_4>0 \mbox{ some constant depending on }T\mbox{ and }\gamma.
    \end{aligned}
    \label{eq:contraction 0}
\end{equation}
\label{th1}
\end{theorem}

\begin{proof}
We define the error vectors and scalars,
\[
    \begin{split}
        &
\veer_{\bcalF}^k :=
    \bcalF\big(\bar\sigma(\veu^{k-1})\,,\,\psi^{k-1},\ddt\veu^{k}\big)-
    \bcalF\big(\bar\sigma(\veu^{k-2})\,,\,\psi^{k-2},\ddt\veu^{k-1}\big), 
\\&
\scer_{\opF}^k :=
    \opF\big(\bar\sigma(\veu^{k-1})\,,\,|T_\frcpm(\ddt\veu^{k})|  \,,\,
    \psi^{k-1}\big)-
    \opF\big(\bar\sigma(\veu^{k-2})\,,\,|T_\frcpm(\ddt\veu^{k-1})|\,,\,
    \psi^{k-2}\big) , 
    \\&
\scer_{\opG}^k :=
    \opG\big(|T_\frcpm(\ddt\veu^{k})|  \,,\,\psi^{k}\big)-
    \opG\big(|T_\frcpm(\ddt\veu^{k-1})|\,,\,\psi^{k-1}\big) \, , \quad
    \\ &
\scer_{s}^k :=
    |T_\frcpm(\ddt\veu^{k})|- |T_\frcpm(\ddt\veu^{k-1})| \,,\quad
\scer_{s}^{\veps,k} :=
    \sqrt{|T_\frcpm(\ddt\veu^k)|^2+\veps^2}-
    \sqrt{|T_\frcpm(\ddt\veu^{k-1})|^2+\veps^2}
    .
\end{split}
\]
It is immediate that 
\begin{equation}
\big|\scer_s^k\big|\leq 
\big|T_\frcpm(\ddt\veer_\veu^k)\big|
=\big|T_{\frc^\pp}(\ddt\veer_\veu^k)_\TT
-T_{\frc^\mm}(\ddt\veer_\veu^k)_\TT\big|
\leq\big|T_{\frc^\pp}(\ddt\veer_\veu^k)\big|
+\big|T_{\frc^\mm}(\ddt\veer_\veu^k)\big|,
    \label{eq:scer s}
\end{equation}
which, using (\ref{eq:trace ineq}), gives
\begin{equation}
    \norm{\scer_s^k}_{L^2(\fSurf)}^2\leq 
    \norm{T_\frcpm(\ddt\veer_\veu^k)}_{L^2(\fSurf)}^2
    \leq\norm{T_\frc(\ddt\veer_\veu^k)}_{L^2(\fSurf)}^2
    \leq C_\frc \norm{\ddt\veer_\veu^k}_{\spH}^2.
    \label{eq:scer s norm}
\end{equation}
It is clear that $\scer_s^k\,\scer_s^{\veps,k}\geq 0$.  
Subtracting iteration $k$ from iteration $k-1$ of (\ref{eq:step 1})
for $k\geq 2$ yields,
\begin{align}
    &
    \begin{aligned}
        &
    \big\langle \rho^0\ddtt\veer_\veu^k \,,\, \vew \big\rangle_{\Domn}
    \hspace{-1mm}
    +a_3'(\veer_\veu^k , \vew)
    -C_{a_3}'\big(\veer_\veu^k , \vew\big)_{L^2(\Domn)}
    -\frac1{4\pi \Grav}\big( \nabla \opS(\veer_\veu^{k-1}) , 
    \nabla \opS(\vew) \big)_{L^2(\mathbb{R}^3)}
    \\ & \hspace{1cm}
    +\gamma \big(\ddt\veer_\veu^k,\vew
    \big)_{\spH}
    +\big\langle\veer_{\bcalF}^{k} \,,\, \vew 
    \big\rangle_{\Surf_\frc}
    -\jmp{\big\langle\vetau_2\big(\veer_{\veu}^{k-1}\big)
    \,,\, \vew \big\rangle_{\Surf_\frc}}
    =0,
    \end{aligned}
    \label{eq:err weak}
\end{align}
where $a_3'$ is defined in \eqref{eq:a3p}.
We let $\vew=\ddt\veer_\veu^{k}$
so that (\ref{eq:err weak}) implies,
\begin{align}
    &
    \begin{aligned}
        & \frac{1}{2} \pddt \Big(
    \norm{\ddt\veer_\veu^k }^2_{L^2(\Domn;\rho^0)}
    +a_3'(\veer_\veu^k,\veer_\veu^k)\Big)
    +\gamma \norm{\ddt\veer_\veu^k }^2_{\spH}
    \\ & \hspace{0.5cm}
    =
    C_{a_3}'\big(\veer_\veu^k , \ddt\veer_{\veu}^k\big)_{L^2{\Domn}}
    +\frac1{4\pi \Grav}\big(
    \nabla \opS(\veer_\veu^{k-1}) \,,\,
    \nabla \opS(\ddt\veer_\veu^k)
    \big)_{L^2(\mathbb{R}^3))}
    \\ & \hspace{1.0cm}
    -\big\langle\veer_{\bcalF}^{k} \,,\,\ddt\veer_\veu^k 
    \big\rangle_{\Surf_\frc}
    +\jmp{\big\langle \vetau_2\big(\veer_{\veu}^{k-1}\big)
    \,,\, \ddt\veer_\veu^k
    \big\rangle_{\Surf_\frc}} .
    \end{aligned}
    \label{eq:err weak 1}
\end{align}
We denote by $I_0,I_1,I_2$ and $I_3$ the three terms on the right-hand
side of (\ref{eq:err weak 1}). From Young's inequality, it follows
that
\begin{equation}
    I_0\leq \frac{C_{a_3}'}2\Big(
    \norm{\veer_\veu^k}_{L^2(\Domn)}^2
    +\norm{\ddt\veer_\veu^k}_{L^2(\Domn)}^2
    \Big).
    \label{eq:I0}
\end{equation}
Based on \cite[p.\ 28 proof of Theorem 2]{deHoop2015}, we have
\[
    \norm{\nabla \opS(\veu)}_{L^2(\mathbb{R}^3)}^2 \leq 
    C_S\norm{\veu}_{L^2(\Domn;\rho^0)}^2,
\]
so that
\begin{equation}
    \begin{split}
        I_1 \leq & 
        \frac1{8\pi \Grav}\Big({\delta_1}
        \norm{\nabla \opS(\veer_\veu^{k-1})}_{L^2(\mathbb{R}^3)}^2
        +\frac1{\delta_1}
        \norm{\nabla \opS(\ddt\veer_\veu^{k})}_{L^2(\mathbb{R}^3)}^2
        \Big)
        \\ \leq &
        \frac{C_S C_{\rho^0}^\star}{8\pi \Grav}\Big({\delta_1}
        \norm{\veer_\veu^{k-1}}_{L^2(\Domn)}^2
        +\frac1{\delta_1}
        \norm{\ddt\veer_\veu^{k}}_{L^2(\Domn)}^2
        \Big).
    \end{split}
    \label{eq:step1 err1}
\end{equation}
Meanwhile,
\begin{equation}
    \begin{split}
        I_2 = & -\int_{\Surf_\frc} \left(
        \opF(\bar\sigma(\veu^{k-1}),|T_\frcpm(\ddt\veu^{k})|,
        \psi^{k-1})
        \left(\frac{|T_\frcpm(\ddt\veu^k)|^2-
        T_\frcpm(\ddt\veu^k)\cdot T_\frcpm(\ddt\veu^{k-1})}
        {\sqrt{|T_\frcpm(\ddt\veu^k)|^2+\veps^2}}\right)\right.
        \\ & \hspace{0.5cm}
        \left.
        +\opF(\bar\sigma(\veu^{k-2}),|T_\frcpm(\ddt\veu^{k-1})|,
        \psi^{k-2})
        \left(\frac{|T_\frcpm(\ddt\veu^{k-1})|^2-
        T_\frcpm(\ddt\veu^k)\cdot T_\frcpm(\ddt\veu^{k-1})}
        {\sqrt{|T_\frcpm(\ddt\veu^{k-1})|^2+\veps^2}}\right)
        \right)\dd\Surf.
    \end{split}
\end{equation} 
To simplify the notation in the algebraic manipulations, we let 
\[
f_1=\opF\big(\bar\sigma(\veu^{k-1}),|T_\frcpm(\ddt\veu^{k})|,
\psi^{k-1}\big),\quad
f_2=\opF\big(\bar\sigma(\veu^{k-2}),|T_\frcpm(\ddt\veu^{k-1})|,
\psi^{k-2}\big)
\]
and
\[
   \vei=T_\frcpm(\ddt\veu^k),\quad
   \vej=T_\frcpm(\ddt\veu^{k-1}), 
\]
whence
\[
I_2=\int_{\Surf_\frc} \left(
f_1\frac{-|\vei|^2+\vei\cdot\vej}
{\sqrt{|\vei|^2+\veps^2}}
+f_2\frac{-|\vej|^2+\vei\cdot\vej}
{\sqrt{|\vej|^2+\veps^2}}\right)
\dd\Surf.
\]
Using the Cauchy-Schwartz inequality,
\[
    \vei\cdot\vej+\veps^2\leq
    \sqrt{(|\vei|^2+\veps^2)(|\vej|^2+\veps^2)},
\]
and it follows that
\begin{equation}
    \begin{split}
        &
f_1\frac{-|\vei|^2+\vei\cdot\vej}
{\sqrt{|\vei|^2+\veps^2}}
+f_2\frac{-|\vej|^2+\vei\cdot\vej}
{\sqrt{|\vej|^2+\veps^2}}
\\&\hspace{1cm}=
f_1\left(-\sqrt{|\vei|^2+\veps^2}
+\frac{\vei\cdot\vej+\veps^2}
{\sqrt{|\vei|^2+\veps^2}}\right)
+f_2\left(-\sqrt{|\vej|^2+\veps^2}
+\frac{\vei\cdot\vej+\veps^2}
{\sqrt{|\vej|^2+\veps^2}}\right)
\\&\hspace{1cm}\leq
(f_1-f_2)(-\sqrt{|\vei|^2+\veps^2}
+\sqrt{|\vej|^2+\veps^2}).
    \end{split}
\end{equation}
We note that
\[
    \big| 
    \sqrt{|\vei|^2+\veps^2}
    -\sqrt{|\vej|^2+\veps^2}\big|
    \leq
    \big| |\vei|-|\vej| \big|,
\]
with the difference going to $0$ uniformly as $\veps$ vanishes.
Hence, $C_\veps|\scer_s^k|\leq |\scer_s^{\veps,k}|\leq
|\scer_s^k|$, with the positive constant $C_\veps \to 1$ as
$\veps \to 0$. Therefore, with the Lipschitz continuity of
$\opF$ expressed by (\ref{eq:bnd F}),
\begin{equation}
    \begin{split}
        I_2  \leq &
    \Cblu
        \int_{\Surf_\frc} \Big(
        \opF\big(\bar\sigma(\veu^{k-1}),|T_\frcpm(\ddt\veu^{k})|,
        \psi^{k-1}\big)
        -\opF\big(\bar\sigma(\veu^{k-2}),|T_\frcpm(\ddt\veu^{k-1})|,
        \psi^{k-2}\big)\Big)
        \\ & \hspace{2cm}
    \Cblu
        \left(\sqrt{|T_\frcpm(\ddt\veu^{k-1})|^2+\veps^2}-
        \sqrt{|T_\frcpm(\ddt\veu^k)|^2+\veps^2}
        \right)\dd\Surf
    \\ = &
        -\int_{\Surf_\frc} 
        \scer_{\opF}^k\,\scer_s^{\veps,k} \dd\Surf
        \approx
        -\int_{\Surf_\frc} \left(
        \pdd{\opF}s\scer_s^k\,\scer_s^{\veps,k} 
        +\pdd{\opF}\sigma \bar\sigma(\veer_\veu^{k-1})\,
        \scer_s^{\veps,k}
        +\pdd{\opF}\psi \scer_\psi^{k-1}\,\scer_s^{\veps,k}
        \right)\dd\Surf
    \\ \leq &
        \int_{\Surf_\frc} \left(
        -C_{\opF,s}C_\veps|\scer_s^k|^2
        +C_{\opF,\sigma}^\star|\bar\sigma(\veer_\veu^{k-1})||\scer_s^k| 
        +C_{\opF,\psi}^\star|\scer_\psi^{k-1}||\scer_s^k| \right)
        \dd\Surf
    \\ \leq &
        - C_{\opF,s}C_\veps\norm{\scer_s^k}_{L^2(\Surf_\frc)}^2
        + C_{\opF,\sigma}^\star\big\langle
        |\bar\sigma(\veer_\veu^{k-1})|
        \,,\, |\scer_s^k|\big\rangle_{\Surf_\frc}
        + C_{\opF,\psi}^\star\big(|\scer_\psi^{k-1}|
        \,,\, |\scer_s^k|\big)_{L^2(\Surf_\frc)}.
    \end{split}
    \label{eq:step1 err2 1}
\end{equation}
Using Lemma~\ref{lm} and then Young's inequality, we obtain
\begin{equation}
    \begin{split}
        &
    \big\langle |\bar\sigma(\veer_\veu^{k-1})|
    \,,\, |\scer_s^k|\big\rangle_{\Surf_\frc}
    =
    \big\langle\big|\ven\cdot\big(
    \itLa_{\tsLaINI,\rho^0,\veg^0}
    +\itLa'_{\tsLaINI,\rho^0,\veg^0}\big)
    \circ T_\frc \,(\veer_\veu^k)\big|\,,\,
    \big|T_\frcpm (\ddt\veer_\veu^k)\big|\big\rangle_{\Surf_\frc}
    \\&\hspace{1cm}
    \leq (C_I+C'_I)\left(\frac1{2\delta_2}
    \norm{\veer_\veu^{k-1}}_{\spH}^2
    +\frac{\delta_2}2\norm{\ddt\veer_\veu^k}_{\spH}^2
    \right).
    \end{split}
    \label{eq:err sigma s}
\end{equation}
Using the Cauchy-Schwarz and Young's inequalities, again,
\begin{equation}
    \begin{split}
    \big(|\scer_\psi^{k-1}|
    \,,\, |\scer_s^k|\big)_{L^2(\Surf_\frc)}
    \leq& \frac1{2\delta_3}
    \norm{\scer_\psi^{k-1}}_{L^2(\Surf_\frc)}^2
    +\frac{\delta_3}2\norm{\scer_s^k}
    _{L^2(\Surf_\frc)}^2.
    \end{split}
    \label{eq:err tau2 s}
\end{equation}
The estimates leading to (\ref{eq:err sigma s}) also lead to
\begin{equation}
    \begin{split}
        &
        I_3= \jmp{\big\langle\vetau_2\big(\veer_{\veu}^{k-1}\big)
        \,,\, \ddt\veer_\veu^k\big\rangle_{\Surf_\frc}}
        \leq \sum_{\pp,\mm}
        \big| \big\langle\itLa'_{\tsLaINI,\rho^0,\veg^0}
        \circ T_\frc \,(\veer_\veu^k)\,,\,
        T_\frc (\ddt\veer_\veu^k)\big\rangle_{\Surf_{\frc^\pm}}\big|
        \\&\hspace{1cm}
        \leq C'_I\left(\frac1{2\delta_4}
        \norm{\veer_\veu^{k-1}}_{\spH}^2
        +\frac{\delta_4}2\norm{\ddt\veer_\veu^k}_{\spH}^2
        \right).
    \end{split}
    \label{eq:err I3}
\end{equation}
We subtract (\ref{eq:weakform ODE}) from (\ref{eq:step 2}) at step
$k$, and let $\varphi=\scer_\psi^{k}$ so that 
\begin{equation}
    \begin{split}
        &
    \frac{1}{2} \pddt\norm{\scer_\psi^{k}}_{L^2(\Surf_\frc)}^2
    =-\big(\scer_{\opG}^k \,,\, \scer_\psi^{k} \big)_{L^2(\Surf_\frc)}
    \leq
    \int_{\Surf_\frc}\left(C_{\opG,s}^\star|\scer_s^{k}|
    -C_{\opG,\psi}|\scer_\psi^{k}| \right)|\scer_\psi^{k}|\dd\Surf
    \\& \hspace{1cm}
    \leq \frac{C_{\opG,s}^\star}2\left(
    \frac1{\delta_5}\norm{\scer_\psi^{k}}_{L^2(\Surf_\frc)}^2
    +\delta_5\norm{\scer_s^k}_{L^2(\Surf_\frc)}^2
    \right)
    -C_{\opG,\psi}\norm{\scer_\psi^k}_{L^2(\Surf_\frc)}^2,
    \end{split}
    \label{eq:err weak ODE}
\end{equation}
in which, based on (\ref{eq:bnd G}), 
\begin{equation}
    \begin{split}
        &
    |\scer_{\opG}^k\,\scer_\psi^k| \approx 
    \left|\pdd{\opG}{s}\scer_s^k\,\scer_\psi^k 
    +\pdd{\opG}{\psi}|\scer_\psi^k|^2 \right|
    \geq 
    -\left|\pdd{\opG}{s}\scer_s^k\,\scer_\psi^k\right| 
    +\pdd{\opG}{\psi} |\scer_\psi^k|^2 
    \geq 
    -C_{\opG,s}^\star |\scer_s^k||\scer_\psi^k| 
    + C_{\opG,\psi}|\scer_\psi^k|^2.
    \end{split}
    \label{eq:step2 err}
\end{equation}
Combining (\ref{eq:err weak 1})-(\ref{eq:err weak ODE}), we get the
estimate
\begin{equation}
    \begin{split}
     & \pddt \Big(
     C_{\rho^0}\norm{\ddt\veer_\veu^k }^2_{L^2(\Domn)}
    +a_3'\big(\veer_\veu^k,\veer_\veu^k\big)
    +\norm{\scer_\psi^{k}}_{L^2(\Surf_\frc)}^2
    \Big)
    \\ & \hspace{0.5cm}
    \leq
    \frac{C_S C_{\rho^0}^{\star}\delta_1}{4\pi \Grav}
    \norm{\veer_\veu^{k-1}}_{L^2(\Domn)}^2
    +\left(
    \frac{C_I+C'_I}{\delta_2}C_{\opF,\sigma}^\star 
    +\frac{C'_I}{\delta_4} \right)
    \norm{\veer_\veu^{k-1}}_{\spH}^2
    + \frac{C_{\opF,\psi}^\star}{\delta_3}
    \norm{\scer_\psi^{k-1}}_{L^2(\Surf_\frc)}^2
    \\
& \hspace{1.0cm}
    +C_{a_3}'
    \norm{\veer_\veu^k}_{L^2(\Domn)}^2
    +\left(C_{a_3}'
    +\frac{C_S C_{\rho^0}^{\star}}{4\pi \Grav\delta_1}\right)
    \norm{\ddt\veer_\veu^{k}}_{L^2(\Domn)}^2
    + \left(\frac{C_{\opG,s}^\star}{\delta_5} -2C_{\opG,\psi}\right)
    \norm{\scer_\psi^k}_{L^2(\Surf_\frc)}^2
    \\
& \hspace{1.0cm}
    +\big((C_I+C'_I)C_{\opF,\sigma}^\star \delta_2
    +C'_I\delta_4 -2\gamma \big)
    \norm{\ddt\veer_\veu^k }^2_{\spH}
    \\[0.25cm]
& \hspace{1.0cm}
    +\big(C_{\opF,\psi}^\star\delta_3
    +C_{\opG,s}^\star\delta_5 - 2 C_\veps C_{\opF,s}\big)
    \norm{\scer_s^k}_{L^2(\Surf_\frc)}^2 .
    \end{split}
    \label{eq:err est}
\end{equation}
We let $\delta_1=1$, 
$(C_I+C'_I) C_{\opF,\sigma}^\star \delta_2=C'_I\delta_4=\tfrac\gamma2$
and 
$C_{\opF,\psi}^\star\delta_3=C_{\opG,s}^\star\delta_5=
C_\veps C_{\opF,s}$,
integrate (\ref{eq:err est}) over $[0,t]$ with $t \leq T$, 
and take \eqref{eq:coercivity a3} into account. Then
\begin{equation}
    \begin{split}
    & 
    C_{\rho^0}\norm{  \ddt\veer_\veu^k }^2_{L^2(\Domn)}
    +C_{a_3}\norm{\veer_\veu^k}^2_{\spH}
    +\norm{\scer_\psi^{k}}_{L^2(\Surf_\frc)}^2
    +\gamma\int_0^t \norm{\ddt\veer_\veu^k}_{\spH}^2\dd\tau
    \\[3mm] & \hspace{0.5cm}
    \leq \int_0^t \Bigg(
    \frac{C_S C_{\rho^0}^{\star}}{4\pi \Grav}
    \norm{\veer_\veu^{k-1}}
    _{L^2(\Domn)}^2
    +\frac2{\gamma}
    ( ((C_I + C'_I) C_{\opF,\sigma}^{\star})^2
      + C_I^{'\,2}) 
    \norm{\veer_\veu^{k-1}}
    _{\spH}^2
    \\[1mm] & \hspace{1.0cm}
    + \frac{C_{\opF,\psi}^{\star\,2}}{C_\veps C_{\opF,s}}
    \norm{\scer_\psi^{k-1}}
    _{L^2(\Surf_\frc)}^2
    +C_{a_3}'
    \norm{\veer_\veu^k}_{L^2(\Domn)}^2
    +\left(C_{a_3}'
    +\frac{C_S C_{\rho^0}^{\star}}{4\pi \Grav}\right)
    \norm{\ddt\veer_\veu^{k}}_{L^2(\Domn)}^2
    \\[1mm] & \hspace{1.0cm}
    + \left(\frac{C_{\opG,s}^{\star\,2}}
    {C_\veps C_{\opF,s}}
    -2C_{\opG,\psi}\right)
    \norm{\scer_\psi^k}
    _{L^2(\Surf_\frc)}^2
    \Bigg)\dd\tau.
    \end{split}
    \label{eq:err est1}
\end{equation}
Applying Gronwall's lemma, \eqref{eq:err est1} results in
\begin{equation}
    \begin{split}
    & 
    \left(C_{\rho^0}-C_{a_3}'\beta(T)
    -\frac{C_S C_{\rho^0}^{\star}\beta(T)}{4\pi \Grav}\right)
    \norm{\ddt\veer_\veu^k }
    ^2_{L^\infty([0,T];L^2(\Domn))}
    +\left(C_{a_3}-\beta(T) C_{a_3}'\right)\norm{\veer_\veu^k}
    ^2_{L^\infty([0,T];\spH)}
    \\ & \hspace{1cm}
    +\,\gamma\, C_T 
    \norm{\ddt\veer_\veu^k}^2_{L^2([0,T];\spH)}
    +\left(1-\bigg(\frac{C_{\opG,s}^{\star\,2}}
    {C_\veps C_{\opF,s}}
    -2C_{\opG,\psi}\bigg)\beta(T)\right)
    \norm{\scer_\psi^{k}}
    ^2_{L^\infty([0,T];L^2(\Surf_\frc))}
    \\[2mm] & \hspace{0.0cm}
    \leq
    \frac{C_S C_{\rho^0}^{\star}\beta(T)}{4\pi \Grav}
    \norm{\veer_\veu^{k-1}}
    ^2_{L^\infty([0,T];L^2(\Domn))}
    +\frac{2\beta(T)}{\gamma}
    \big(( (C_I+C'_I)C_{\opF,\sigma}^{\star})^2
    +C_I^{'\,2}\big) 
    \norm{\veer_\veu^{k-1}}
    ^2_{L^\infty([0,T];\spH)}
    \\[2mm] & \hspace{1cm}
    + \frac{C_{\opF,\psi}^{\star\,2}\beta(T)}{C_\veps C_{\opF,s}}
    \norm{\scer_\psi^{k-1}}
    ^2_{L^\infty([0,T];L^2(\Surf_\frc))},
    \end{split}
    \label{eq:err est2}
\end{equation}
where $\beta(T) > 0$ is a constant monotonically increasing with the
length of the time interval $[0,T]$. Clearly, the contraction
\eqref{eq:contraction 1} is a direct consequence of \eqref{eq:err
  est2} if the criteria in (\ref{eq:contr cond}) are satisfied. We,
therefore, obtain a contraction within $ V_1 \times V_2$.
\end{proof}

\begin{remark}
    To properly control the error, the time interval $[0,T]$ must be
    sufficiently small such that $\beta(T)$ satisfies the condition in
    \eqref{eq:contr cond}.  For a long-time simulation, the overall
    time is subdivided into sufficiently small time intervals, namely,
    \[
        [0,\Dt],[\Dt,2\Dt],[2\Dt,3\Dt],\cdots,[(N-1)\Dt,N\Dt],
        \quad \Dt:=T/N,
    \]
    and iterations are conducted within each time segment. In this
    way, a sufficiently small $\beta(\Dt)$ can be used alternatively
    in Theorem \ref{th1}.
\end{remark}

\begin{remark}
    As is clear from \eqref{eq:contr cond}, a vanishing $\gamma$
    destroys the contraction. This has been apparent in computational
    experiments: A zero $\gamma$ results in severe oscillations
    \cite{Martins1987} while a positive $\gamma$ aids the stability of
    numerical simultions. Indeed, it is necessary to make sure that
    $\gamma$ takes a strictly positive value to be able to apply the
    Kelvin-Voigt regularization.
\end{remark}

\section{Existence of a weak solution}
\label{sec:exist}

Here, we follow the method of proof by Martins and Oden (1987)
\cite{Martins1987}. We immediately obtain a corollary as a
  consequence of the contraction of the sequence shown in
  \eqref{eq:contraction 1}, and then a theorem that guarantees the
  existence of a solution to Problem \ref{pb1}.

\begin{corollary}
    Suppose $(\veu^k,\psi^k)\in V_1\times V_2$ are sequence of 
    solutions to the scheme described by Problem \ref{pb steps},
    then there exists $(\veu,\psi)\in V_1 \times V_2$ such that
    \[
        \begin{array}{rclrl}
        \veu^k &\,\to &\,\veu 
        &\,\mbox{ strongly in } &\,
        L^\infty\big([0,T];L^2(\Domn)\big);
        \\[2mm]
        \nabla\veu^k &\,\to &\,\nabla\veu 
        &\,\mbox{ strongly in } &\,
        L^\infty\big([0,T];L^2(\Domn)\big);
        \\[2mm]
        \ddt\veu^k &\,\to &\,\ddt\veu 
        &\,\mbox{ strongly in } &\,
        L^\infty\big([0,T];L^2(\Domn)\big);
        \\[2mm]
        \nabla\ddt\veu^k &\,\to &\,\nabla\ddt\veu 
        &\,\mbox{ strongly in } &\,
        L^2\big([0,T];L^2(\Domn)\big);
        \\[2mm]
        \psi^k &\,\to &\,\psi 
        &\,\mbox{ strongly in } &\,
        L^\infty \big([0,T];L^2(\Surf_\frc)\big).
        \end{array}
    \]
    \label{co1}
\end{corollary}

\begin{theorem}
Let the assumptions in Theorem \ref{th1} hold true. Then the limit of
the sequence written in Corollary \ref{co1}, $(\veu,\psi)\in V_1\times
V_2$, solves Problem \ref{pb1}.
\label{th:exist}
\end{theorem}

\begin{proof}
Due to the fact that the trace operator $T_\frcpm:\spH\to
H^{1/2}(\Surf_\frc)$ is Lipschitz continuous, we have the convergence,
\begin{equation}
    T_\frcpm(\ddt\veu^k) \to T_\frcpm(\ddt\veu) 
    \mbox{ strongly in } 
    L^2 \big([0,T];L^2(\Surf_\frc)\big).
    \label{eq: trace conv}
\end{equation}
Meanwhile, $\opG $ is a function of two variables, and is Lipschitz 
with regard to each, that is, for a fixed $\psi^k$,
\[
    \begin{split}
        &
    \big|\opG\big(|T_\frcpm(\ddt\veu^k)|\,,\, \psi^k\big) -
         \opG\big(|T_\frcpm(\ddt\veu)|  \,,\, \psi^k\big) \big|
    \leq c_1 \big||T_\frcpm(\ddt\veu^k)|-|T_\frcpm(\ddt\veu)|\big|
    \\&\hspace{1cm}
    \leq c_1 \big|T_\frcpm(\ddt\veu^k)-T_\frcpm(\ddt\veu)\big|,
    \end{split}
\]
where the second inequality follows from the triangle inequality, 
and for a fixed $\ddt\veu$, 
\[
    \begin{split}
        &
    \big|\opG\big(|T_\frcpm(\ddt\veu)|\,,\,\psi^k\big) -
         \opG\big(|T_\frcpm(\ddt\veu)|\,,\,\psi\big) \big|
    \leq c_2 \big|\psi^k-\psi\big|.
    \end{split}
\]
Therefore,
\begin{equation}
    \begin{split}
        &
    \norm{\opG\big(|T_\frcpm(\ddt\veu^k)|\,,\,\psi^k\big) -
          \opG\big(|T_\frcpm(\ddt\veu)|  \,,\,\psi\big) 
    }_{L^2(\Surf_\frc)}
    \\&\hspace{0.5cm}
    \leq 
    \norm{\opG\big(|T_\frcpm(\ddt\veu^k)|\,,\,\psi^k\big) -
          \opG\big(|T_\frcpm(\ddt\veu)|  \,,\,\psi^k\big) 
    }_{L^2(\Surf_\frc)}
    \\&\hspace{1cm}
    +\norm{\opG\big(|T_\frcpm(\ddt\veu)|\,,\,\psi^k\big) -
           \opG\big(|T_\frcpm(\ddt\veu)|\,,\,\psi\big)
    }_{L^2(\Surf_\frc)}
    \\&\hspace{0.5cm}
    \leq c_1 \norm{T_\frcpm(\ddt\veu^k)-T_\frcpm(\ddt\veu)}
    _{L^2(\Surf_\frc)}
    + c_2 \norm{\psi^k-\psi}_{L^2(\Surf_\frc)}
    \\&\hspace{0.5cm}
    \leq C_1 \norm{\ddt\veu^k-\ddt\veu}
    _{\spH}
    + c_2 \norm{\psi^k-\psi}_{L^2(\Surf_\frc)}.
    \end{split}
    \label{eq: G conv}
\end{equation}
We invoke the usual density argument, take $\phi$ sufficiently smooth,
with the property $\phi|_{t=T}=0$, do integration by parts, and pass
the limit $\psi^k\to\psi$,
\[
    \begin{aligned}
        &
    \int_0^T\int_{\Surf_\frc}
    \ddt\psi^k\,\varphi
    \dd\Surf\dd t
    =
    -\int_0^T\int_{\Surf_\frc}
    \psi^k\,\ddt\varphi
    \dd\Surf\dd t
    \\&\hspace{4cm}
    \xrightarrow{k\to \infty}
    -\int_0^T\int_{\Surf_\frc}
    \psi\,\ddt\varphi
    \dd\Surf\dd t
    =
    \int_0^T\int_{\Surf_\frc}
    \ddt\psi\,\varphi
    \dd\Surf\dd t.
    \end{aligned}
\]
Taking for $\phi$ a product of smooth functions of space and time,
from \eqref{eq: trace conv}--\eqref{eq: G conv} we conclude that
\begin{equation}
    \int_{\Surf_\frc}
    \opG\big(|T_\frcpm(\ddt\veu^k)|\,,\,\psi^k\big)\,\varphi\dd\Surf
    \xrightarrow{k\to \infty}
    \int_{\Surf_\frc}\opG\big(|T_\frcpm(\ddt\veu)|\,,\,\psi\big)
    \,\varphi\dd\Surf
    \quad\mbox{in } L^2\big([0,T]\big).
    \label{eq: conv G}
\end{equation}
Therefore, $\ddt\psi\in L^2\big([0,T];L^2(\Surf_\frc)\big)$, and the
pair $(\veu,\psi)$ solves \eqref{eq:step 2}.

Now, we consider \eqref{eq:step 1}. First of all, a direct
consequence of Corollary \ref{co1} is that
\begin{equation}
    \big(\ddt\veu^k,\,\vew\big)_\spH
    \xrightarrow{k\to \infty}
    \big(\ddt\veu,\,\vew\big)_\spH
    \quad\mbox{in } L^2\big([0,T]\big).
    \label{eq:conv gamma H}
\end{equation}
Secondly, since $\opF$ is a function of 
three variables and is Lipschitz continuous with regard to each
(and so is $\bcalF$), following the same procedure for obtaining
\eqref{eq: conv G}, we can also get
\begin{equation}
    \big\langle
    \bcalF\big(\sigma^{k-1},\ddt\veu^{k},\psi^{k-1}\big),
    \vew \big\rangle_{\Surf_\frc}
    \xrightarrow{k\to \infty}
    \big\langle
    \bcalF\big(\sigma,\ddt\veu,\psi\big),
    \vew \big\rangle_{\Surf_\frc}
    \quad\mbox{in } L^2\big([0,T]\big).
    \label{}
\end{equation}
Meanwhile, $a_3(\veu,\vew)$ is a bilinear form that contains
$\nabla\veu_k$ and $\veu_k$. From Corollary \ref{co1}, 
we conclude that
\begin{equation}
    a_3\big(\veu^k\,,\,\vew\big) \to a_3\big(\veu\,,\,\vew\big)
    \quad\mbox{in } L^\infty\big([0,T]\big).
    \label{}
\end{equation}
Also, the linear map $\nabla \opS:L^2(\Domn) \to L^2(\mathbb{R}^3)$ is a 
Lipschitz continuous map, which indicates that
\begin{equation}
    \norm{\nabla \opS(\veu^k)-\nabla \opS(\veu)}_{L^2(\mathbb{R}^3)}
    \leq C \norm{\veu^k-\veu}_{L^2(\Domn)}
    \xrightarrow{k\to \infty}
    0 
    \quad\mbox{in } L^\infty\big([0,T]\big).
    \label{}
\end{equation}
Finally, with $\vetau_2=\itLa'_{\tsLaINI,\rho^0,\veg^0}\circ T_\frc$ a
linear map, where $\itLa'_{\tsLaINI,\rho^0,\veg^0}$ is the bounded
linear Dirichlet-to-Neumann map $H^{1/2}(\Surf_\frc) \to
H^{-1/2}(\Surf_\frc)$, we therefore have
\begin{equation}
    \jmp{\big\langle\vetau_2(\veu^{k-1}) \,,\, \vew 
    \big\rangle_{\Surf_\frc}}
    \xrightarrow{k\to \infty}
    \jmp{\big\langle\vetau_2(\veu) \,,\, \vew 
    \big\rangle_{\Surf_\frc}}
    \quad\mbox{in } L^2\big([0,T]\big).
    \label{eq:conv D2N}
\end{equation}
Based on the discussion above, we find that all the terms in
\eqref{eq:weakform}, except the one containing $\ddtt \veu$, are well
defined in the space of $L^2([0,T])$. It is then clear that
$\big\langle \rho^0\ddtt\veu^{k}\,,\, \vew \big\rangle_{\Domn}$ is
well defined in the distribution space $\mathcal{D}'([0,T])$. We move
all terms except the one containing $\ddtt \veu^k$ in \eqref{eq:step
  1} into right-hand-side, the result of which can be written in the
formal way,
\begin{equation}
    \big\langle \rho^0\ddtt\veu^k \,,\, \vew \big\rangle_{\Domn}
    =\mathcal{A}\big(\veu^k,\ddt\veu^k;\vew \big),
    \label{}
\end{equation}
where $\mathcal{A}$ is nonlinear with regard to $\veu^k$ and
$\ddt\veu^k$, and is linear with regard to $\vew$. By a density
argument, for any $\vew \in L^2\big([0,T];\spH\big)$ there exist a
sequence $\big\{\vew_i\big\}_{i=1}^\infty$ with $\vew_i \in
\mathcal{C}^\infty\big([0,T];\spH\big)$, such that
$\vew_i\xrightarrow{i\to \infty}\vew $. Thus
\begin{equation}
    \begin{aligned}
        &
    \int_0^T
    \big\langle \rho^0\ddtt\veu^{k}\,,\,
    \vew_i \big\rangle_{\Domn}
    \dd \tau
    =
    \int_0^T
    \mathcal{A}\big(\veu^k,\ddt\veu^k;\vew_i \big)
    \dd \tau
    \\&\hspace{3cm}
    \xrightarrow{k\to \infty}
    \int_0^T
    \mathcal{A}\big(\veu,\ddt\veu;\vew_i \big)
    \dd \tau
    =
    \int_0^T
    \big\langle \rho^0\ddtt\veu\,,\,
    \vew_i \big\rangle_{\Domn}.
    \end{aligned}
    \label{}
\end{equation}
As $i\to \infty$, the following limit holds
\[
    \int_0^T
    \big\langle \rho^0\ddtt\veu^{k}\,,\,
    \vew \big\rangle_{\Domn}
    \dd \tau
    \xrightarrow{k\to \infty}
    \int_0^T
    \big\langle \rho^0\ddtt\veu\,,\,
    \vew \big\rangle_{\Domn},\quad
    \mbox{ for all } \vew \in 
    L^2\big([0,T];\spH\big).
\]
Therefore, $\ddtt\veu \in L^2\big([0,T];\spH'\big)$ (see also
\cite{Martins1987}). We can now conclude that $(\veu,\psi)\in
V_1\times V_2$ also solves \eqref{eq:step 1}, and therefore the
coupled system.
\end{proof}

\section{Implicit discretization in time}
\label{sec:disc t}

We define the space of time-discretized weak solutions for Problem
\ref{pb2} below as $\hat V_1\times \hat V_2 $, with
\begin{equation}
    \begin{split}
        \hat V_1:=&\left\{
            \hat \vev \in 
            \spH
            \,\Big\vert\,
            \jmp{\ven\cdot\hat \vev}=0
            \mbox{ on } \Surf_\frc 
        \right\},
        \\
        \hat V_2\,=&\,
        L^2(\Surf_\frc).
    \end{split}
    \label{eq:space of disc sol}
\end{equation}
To simplify the analysis, we semi-discretize Problem \ref{pb1} with a
uniform time step. We denote the particle velocity $\vev:=\ddt\veu$,
and discretize the time interval by $\Dt=\frac TN$, and let
$t_n=n\Dt$. We use index $n$ in the superscript $\hat v^\br{n}$ to
indicate the time-discretized solution to a time dependent variable
$v$ at time $t=t_n$. A backward Euler time discretization of Problem
\ref{pb1} is described by the following scheme.

\medskip\medskip

\begin{problem}
    Let $\veps$ and $\gamma$ be fixed, strictly positive constants,
    and let the solution for the previous time step $t=t_{n-1}$, $
    (\hat\veu^\br{n-1},\hat\vev^\br{n-1},\hat\psi^\br{n-1}) \in \hat
    V_1\times \hat V_1\times \hat V_2, $ be given. Find the solution
    $(\hat\veu^\br{n},\hat\vev^\br{n},\hat\psi^\br{n}) \in \hat
    V_1\times \hat V_1\times \hat V_2$ for the current time step
    $t=t_n$, such that
\begin{subequations}\label{eq:disc}
\begin{align}
        &
    \begin{aligned}
        &
    \frac1\Dt
    \innp{ \rho^0\,\hat\vev^\br{n}\,,\,\hat\vew }_{L^2(\Domn)}
    +a_3\big(\hat\veu^\br{n}\,,\,\hat\vew\big)
    -\frac1{4\pi \Grav}
    \innp{ \nabla \opS(\hat\veu^\br{n})\,,\, \nabla \opS(\hat\vew) }
    _{L^2(\mathbb{R}^3)}
    \\[1mm] & \hspace{5mm}
    +\gamma \innp{\hat\vev^\br{n} \,,\, \hat\vew }_{\spH}
    +\big\langle\bcalF\big(\bar\sigma(\hat\veu^\br{n}),\hat\psi^\br{n},
    \hat\vev^\br{n}\big) \,,\, \hat\vew 
    \big\rangle_{\Surf_\frc}
    -\jmp{\big\langle\vetau_2(\hat\veu^\br{n}) \,,\, \hat\vew
    \big\rangle_{\Surf_\frc} }
    \\[1mm] & \hspace{5mm}
    =
    \big\langle\ven\cdot(\tsT^0+\tsTd^\br{n}) \,,\, T_\frcpm(\hat\vew) 
    \big\rangle_{\Surf_\frc} 
    +\frac1\Dt
    \innp{ \rho^0\,\hat\vev^\br{n-1}\,,\,\hat\vew }_{L^2(\Domn)}
    ,
    \end{aligned}
    \label{eq:disc v}
    \\[2mm] &
    \hat\veu^\br{n}-\Dt\hat\vev^\br{n}=\hat\veu^\br{n-1} ,
    \label{eq:disc u}
    \\[2mm] &
    \frac1\Dt\innp{\hat\psi^\br{n} \,,\,\hat\varphi}_{L^2(\Surf_\frc)}
    +
    \innp{\opG\big(|T_\frcpm(\hat\vev^\br{n})|,\hat\psi^\br{n}\big),\hat\varphi}
    _{L^2(\Surf_\frc)}
    = 
    \frac1\Dt\innp{\hat\psi^\br{n-1}\,,\,\hat\varphi}_{L^2(\Surf_\frc)}
    \label{eq:disc psi}
\end{align}
\end{subequations}
hold for all $(\hat\vew,\,\hat\varphi) \in \hat V_1 \times \hat V_2$.
\label{pb2}
\end{problem}

\medskip\medskip

\noindent
The corresponding iterative coupling scheme to Problem \ref{pb2}
is similar to the one in (\ref{eq:step 1}-\ref{eq:step 2}) 

\medskip\medskip

\begin{problem}
Let $\veps$ and $\gamma$ be fixed, strictly positive constants, and
let the solutions,
$(\hat\veu^\br{n-1},\hat\vev^\br{n-1},\hat\psi^\br{n-1}) \in \hat
V_1\times \hat V_1\times \hat V_2$ for the previous time step
$t=t_{n-1}$, be given. Assume that
$(\hat\veu^\br{n,k-1},\hat\vev^\br{n,k-1},\hat\psi^\br{n,k-1}) \in
\hat V_1\times \hat V_1\times \hat V_2$ is the solution for the
current time step $t=t_n$ at iteration $k-1$ with $k\geq 1$, where at
the beginning take
\[
\hat\veu^\br{n,0}=\hat\veu^\br{n-1},\quad
\hat\vev^\br{n,0}=\hat\vev^\br{n-1},\quad\mbox{ and }\quad
\hat\psi^\br{n,0}=\hat\psi^\br{n-1}.
\]
Find the solution
$(\hat\veu^\br{n,k},\hat\vev^\br{n,k},\hat\psi^\br{n,k}) \in \hat
V_1\times \hat V_1\times \hat V_2$ for $t=t_n$ at iteration $k$, such
that
\begin{subequations}\label{eq:disc split}
    \begin{align}
        &
    \begin{aligned}
        &
    \frac1\Dt
    \innp{ \rho^0\,\hat\vev^\br{n,k}\,,\,\hat\vew }_{L^2(\Domn)}
    +a_3'\big(\hat\veu^\br{n,k}\,,\,\hat\vew\big)
    -C_{a_3}'\innp{\hat\veu^{\br{n,k}}\,,\,\hat\vew}_{L^2(\Domn)}
    \\[1mm] & \hspace{5mm}
    -\frac1{4\pi \Grav}
    \innp{\nabla \opS(\hat\veu^{\br{n,k-1}})\,,\,\nabla \opS(\hat\vew)}
    _{L^2(\mathbb{R}^3)}
    +\gamma\innp{\hat\vev^{\br{n,k}}\,,\,\hat\vew}
    _{\spH}
    \\[1mm] & \hspace{5mm}
    +\big\langle\bcalF\big(\bar\sigma(\hat\veu^{\br{n,k-1}}),
    \hat\psi^{\br{n,k-1}},\hat\vev^{\br{n,k}}\big)
    \,,\,\hat\vew\big\rangle
    _{\Surf_\frc}
    -\jmp{\big\langle\vetau_2(\hat\veu^{\br{n,k-1}})\,,\,\hat\vew
    \big\rangle_{\Surf_\frc}}
    \\[1mm] & \hspace{5mm}
    =
    \big\langle\ven\cdot(\tsT^0+\tsTd^\br{n})\,,\,T_\frcpm(\hat\vew)
    \big\rangle_{\Surf_\frc} 
    +\frac1\Dt
    \innp{ \rho^0\,\hat\vev^\br{n-1}\,,\,\hat\vew }_{L^2(\Domn)}
    ,
    \end{aligned}\label{eq:disc split v}
    \\[2mm] &
    \hat\veu^{\br{n,k}}-\Dt\hat\vev^{\br{n,k}}=\hat\veu^\br{n-1} ,
    \label{eq:disc split u}
    \\[1mm]&
    \frac1\Dt\innp{\hat\psi^{\br{n,k}}\,,\,\hat\varphi}
    _{L^2(\Surf_\frc)}
    +
    \innp{\opG\big(|T_\frcpm(\hat\vev^{\br{n,k}})|,\hat\psi^{\br{n,k}}\big),
    \hat\varphi}
    _{L^2(\Surf_\frc)}
    =
    \frac1\Dt\innp{\hat\psi^\br{n-1}\,,\,\hat\varphi}
    _{L^2(\Surf_\frc)}
    \label{eq:disc split psi}
    \end{align}
\end{subequations}
hold for all $(\hat\vew,\,\hat\varphi) \in \hat V_1 \times \hat V_2$.
\label{pb3}
\end{problem}

\medskip\medskip

\noindent 
In the remainder of this section, subject to existence of a
time-continuous solution in $V_1\times V_2$ that is given in Theorem
\ref{th:exist}, we prove that the solution of Problem~\ref{pb3}
linearly converges to the unique solution of Problem~\ref{pb2} under
some restrictions on the time step as well as the viscosity
coefficient, by a given rate $\lambda\in (0,1)$.

\begin{theorem}
Let $\gamma$ and $\Dt$ satisfy
\begin{equation}
    \begin{split}
    \frac1{\Dt}
    \geq &
    \max\left(
    \frac{C_{\opF,\psi}^{\star\,2}}{2\lambda C_\veps C_{\opF,s}}
    +\frac{C_{\opG,s}^{\star\,2}}{2C_\veps C_{\opF,s}}
    -C_{\opG,\psi}
    \,,\,
    \sqrt{\left(
    \frac{C_S C_{\rho^0}^\star}{4\pi \Grav C_{\rho^0}\sqrt{\lambda}}
    +\frac{C_{a_3}'}{C_{\rho^0}}
    \right)}
    \,\right),
\\
    \frac\gamma \Dt
    \geq &
    \frac{\big(C_{\opF,\sigma}^{\star\,2}(C_I+C'_I)^2
    +C^{'\,2}_I\big)}{2\lambda C_{a_3}}
        \end{split}
        \label{eq:disc contr cond}
\end{equation}
for some constant $\lambda \in (0,1)$. Then the iterative coupling
scheme described by Problem \ref{pb3} is a contraction in the sense
that
\begin{equation}
    \begin{aligned}
        &
      \kappa_1\norm{\Veer_\vev^k}^2_{L^2(\Domn)}
    + \kappa_2\norm{\Veer_\veu^k}^2_{\spH}
    + \kappa_3\norm{\Scer_\psi^k}^2_{L^2(\Surf_\frc)}
    \\&\hspace{1cm}
    \leq\lambda\left(
      \kappa_1\norm{\Veer_\vev^{k-1}}^2_{L^2(\Domn)}
    + \kappa_2\norm{\Veer_\veu^{k-1}}^2_{\spH}
    + \kappa_3\norm{\Scer_\psi^{k-1}}^2_{L^2(\Surf_\frc)}
    \right)
    ,
    \end{aligned}
    \label{eq:contraction 2}
\end{equation}
where
\[
\Veer_\veu^k:=\hat\veu^{\br{n,k}}-\hat\veu^\br{n,k-1}, \,
\Veer_\vev^k:=\hat\vev^{\br{n,k}}-\hat\vev^\br{n,k-1}, \,
\Scer_\psi^k:=\hat\psi^\br{n,k}-\hat\psi^\br{n,k-1}
\]
and
\[
    \kappa_1=
    \frac{C_{\rho^0}}\Dt - \Dt C_{a_3}'
    -\frac{\Dt C_S C_{\rho^0}^\star}{8\,\pi \Grav\sqrt{\lambda}}
    ,\quad
    \kappa_2=
    \Dt^{-1}C_{a_3}
    \quad
    \mbox{and }\hspace{1mm}
    \kappa_3=
    \frac1{\Dt}
    -\frac{C_{\opG,s}^{\star\,2}}{2C_\veps C_{\opF,s}}
    \quad.
\]
\label{th2}
\end{theorem}

\begin{proof}
We define the error vectors and scalars,
\[
\begin{array}{rcl}
\Veer_{\bcalF}^{k} &:=&
\bcalF\big(\bar\sigma(\hat\veu^{\br{n,k-1}}),\hat\psi^{\br{n,k-1}},
\hat\vev^{\br{n,k}}\big)
\\[1mm] & &\hspace{1cm}
-\bcalF\big(\bar\sigma(\hat\veu^\br{n,k-2}),\hat\psi^\br{n,k-2},
\hat\vev^\br{n,k-1}\big) ,
\\[3mm]
\Scer_{\opF}^k &:=&
\opF\big(\bar\sigma(\hat\veu^{\br{n,k-1}}),
|T_\frcpm(\hat\vev^{\br{n,k}})|,\hat\psi^{\br{n,k-1}}\big)
\\[1mm] & &\hspace{1cm}
-\opF\big(\bar\sigma(\hat\veu^\br{n,k-2}),
|T_\frcpm(\hat\vev^\br{n,k-1})|,\hat\psi^\br{n,k-2}\big) ,
\\[3mm]
\Scer_{\opG}^k &:=&
\opG\big(|T_\frcpm(\hat\vev^\br{n,k})|,\hat\psi^\br{n,k}\big)-
\opG\big(|T_\frcpm(\hat\vev^\br{n,k-1})|,\hat\psi^\br{n,k-1}\big) ,
\\[3mm]
\Scer_{s}^k &:=&
|T_\frcpm(\hat\veu^{\br{n,k}})|-|T_\frcpm(\hat\veu^\br{n,k-1})| .
\end{array}
\]
Clearly, $\Veer_\veu^k=\Dt\Veer_\vev^k$ based on \eqref{eq:disc split u},
and similarly to (\ref{eq:scer s norm}), 
\begin{equation}
    \norm{\Scer_s^k}_{L^2(\fSurf)}^2\leq 
    \norm{T_\frcpm(\Veer_\vev^k)}_{L^2(\fSurf)}^2
    \leq\norm{T_\frc(\Veer_\vev^k)}_{L^2(\fSurf)}^2
    \leq C_\frc \norm{\Veer_\vev^k}_{\spH}^2.
    \label{eq:scer s norm disc}
\end{equation}
We subtract (\ref{eq:disc split}a--c) at iteration $k-1$ from
  the corresponding equations at iteration $k$ to obtain the error
  estimate,
\begin{align}
    &
    \begin{aligned}
        &
    \frac1\Dt
    \innp{ \rho^0\,\Veer_\vev^\br{k}\,,\,\hat\vew }_{L^2(\Domn)}
    +\Dt a_3'\big(\Veer_\vev^k\,,\,\hat\vew\big)
    -\Dt C_{a_3}'\innp{\Veer_\vev^k\,,\,\hat\vew}_{L^2(\Domn)}
    \\[1mm] & \hspace{1cm}
    -\frac\Dt{4\pi \Grav} \big( \nabla \opS(\Veer_\vev^{k-1})
    \,,\, \nabla \opS(\hat\vew) \big)_{L^2(\mathbb{R}^3)}
    +\gamma \big(\Veer_\vev^{k} \,,\, \hat\vew \big)_{\spH}
    \\[1mm] & \hspace{1cm}
    +\big\langle\Veer_{\bcalF}^{k}
    \,,\, \hat\vew \big\rangle_{\Surf_\frc}
    -\jmp{\big\langle\vetau_2\big(\Veer_{\veu}^{k-1} \big)
    \,,\,\hat\vew\big\rangle_{\Surf_\frc}}
    =0,
    \end{aligned}
    \label{eq:err v}
    \\[2mm]&
    \vetau_2\big(\Veer_{\veu}^{k}\big)=-\Dt
    \itLa'_{\tsLaINI,\rho^0,\veg^0}
    (\Veer_\vev^k),
    \label{eq:err tau2}
    \\[1mm]&
    \bar\sigma\big(\Veer_{\veu}^{k}\big)=-\Dt\ven\cdot\big(
    \itLa_{\tsLaINI,\rho^0,\veg^0}+
    \itLa'_{\tsLaINI,\rho^0,\veg^0}\big)
    (\Veer_\vev^k).
    \label{eq:err sigma}
\end{align}
We let $\hat\vew=\Veer_\vev^{k}$, so that (\ref{eq:err v}) becomes
\begin{equation}
    \begin{aligned}
        &
    \frac1\Dt\norm{\Veer_\vev^{k} }^2_{L^2(\Domn;\rho^0)}
    +\Dt a_3'\big(\Veer_\vev^{k},\Veer_\vev^k\big)
    +\gamma \norm{\Veer_\vev^{k}}^2_{\spH}
    \\[0.5mm] & \hspace{0.5cm}
    =
    \Dt C_{a_3}'\norm{\Veer_\vev^k}_{L^2(\Domn)}
    +\frac\Dt{4\pi \Grav} \big( \nabla \opS(\Veer_\vev^{k-1})
    \,,\, \nabla \opS(\Veer_\vev^{k}) \big)_{L^2(\mathbb{R}^3)}
    \\[0.5mm] & \hspace{1cm}
    -\big\langle\Veer_{\bcalF}^{k}
    \,,\, \Veer_\vev^k \big\rangle_{\Surf_\frc}
    +\jmp{\big\langle\vetau_2\big(\Veer_{\veu}^{k-1} \big)
    \,,\,\Veer_\vev^k\big\rangle_{\Surf_\frc}}.
    \end{aligned}
    \label{eq:err v1}
\end{equation}
We denote by $J_0, J_1, J_2$ and $J_3$ the terms on the right-hand
side of (\ref{eq:err v1}), and similarly to (\ref{eq:step1
  err1}-\ref{eq:err I3}),
\begin{align}
    J_1 \leq & 
    \frac{\Dt C_S C_{\rho^0}^\star}{8\pi \Grav}\left(\frac1{\delta_6}
    \norm{\Veer_\vev^{k-1}}_{L^2(\Domn)}^2
    +{\delta_6}
    \norm{\Veer_\vev^{k}}_{L^2(\Domn)}^2
    \right),
    \label{eq:err J1}
    \\[2mm]
    J_2 \leq &
    - C_\veps C_{\opF,s}\norm{\Scer_s^k}_{L^2(\Surf_\frc)}^2
    + C_{\opF,\sigma}^\star\big\langle|\bar\sigma(\Veer_\veu^{k-1})|
    \,,\, |\Scer_s^k|\big\rangle_{\Surf_\frc}
    + C_{\opF,\psi}^\star\big(|\Scer_\psi^{k-1}|
    \,,\, |\Scer_s^k|\big)_{L^2(\Surf_\frc)},
    \label{eq:err J2}
\end{align}
with
\begin{align}
    \big\langle |\bar\sigma(\Veer_\veu^{k-1})|
    \,,\, |\Scer_s^k|\big\rangle_{\Surf_\frc}
    \leq&\, 
    \Dt (C_I+C'_I) \left(\frac1{2\delta_7}
    \norm{\Veer_\vev^{k-1}}_{\spH}^2
    +\frac{\delta_7}2\norm{\Veer_\vev^k}_{\spH}^2
    \right),
    \\
    \big(|\Scer_\psi^{k-1}|
    \,,\, |\Scer_s^k|\big)_{L^2(\Surf_\frc)}
    \leq&\, 
    \left( \frac1{2\delta_3}
    \norm{\Scer_\psi^{k-1}}_{L^2(\Surf_\frc)}^2
    +\frac{\delta_3}2\norm{\Scer_s^k}_{L^2(\Surf_\frc)}^2
    \right)
\end{align}
and
\begin{equation}
    J_3
    \leq \Dt C'_I\left(\frac1{2\delta_8}
    \norm{\Veer_\vev^{k-1}}_{\spH}^2
    +\frac{\delta_8}2\norm{\Veer_\vev^k}_{\spH}^2
    \right).
    \label{eq:err J3}
\end{equation}
We also subtract (\ref{eq:disc split psi}) at step $k-1$ from
  the corresponding equation at step $k$, let
$\varphi=\Scer_\psi^{k}$, and obtain the estimate
\begin{equation}
    \begin{split}
    &
    \frac1{\Dt}
    \norm{\Scer_\psi^{k}}_{L^2(\Surf_\frc)}^2
    =\big(\Scer_{\opG}^k \,,\, \Scer_\psi^{k} \big)_{L^2(\Surf_\frc)}
    \leq
    \int_{\Surf_\frc}\big(C_{\opG,s}^\star|\Scer_s^{k}|
    -C_{\opG,\psi}|\Scer_\psi^{k}| \big)|\Scer_\psi^{k}|\dd\Surf
    \\& \hspace{1cm}
    \leq \frac{C_{\opG,s}^\star}2\left(
    \frac1{\delta_5}\norm{\Scer_\psi^{k}}_{L^2(\Surf_\frc)}^2
    +\delta_5\norm{\Scer_s^k}_{L^2(\Surf_\frc)}^2
    \right)
    -C_{\opG,\psi}\norm{\Scer_\psi^k}_{L^2(\Surf_\frc)}^2.
    \end{split}
    \label{eq:err disc ODE}
\end{equation}
We use the relation $\Veer_\veu^k=\Dt\Veer_\vev^k$, and combine
(\ref{eq:err v1})-(\ref{eq:err disc ODE}) to obtain
\begin{equation}
    \begin{aligned}
        &
    \left(\frac{C_{\rho^0}}\Dt - \Dt C_{a_3}'
    -\frac{\Dt C_S C_{\rho^0}^\star\delta_6}{8\,\pi \Grav}
    \right)
    \norm{\Veer_\vev^k}_{L^2(\Domn)}^2
    +\Dt C_{a_3}
    \norm{\Veer_\vev^k}_{\spH}^2
    \\[1mm] & \hspace{0.5cm}
    +\left(
    \frac\gamma{\Dt^2}
    -\frac{\delta_7\, C_{\opF,\sigma}^\star (C_I+C'_I)}{2\Dt}
    -\frac{\delta_8 \,C'_I}{2\Dt} \right)
    \norm{\Veer_\veu^k}_{\spH}^2
    \\[1mm] & \hspace{0.5cm}
    +\left(\frac1{\Dt}
    -\frac{C_{\opG,s}^\star}{2\delta_5}
    +C_{\opG,\psi}\right)
    \norm{\Scer_\psi^k}_{L^2(\Surf_\frc)}^2
    + \left(C_\veps C_{\opF,s}
    -\frac{\delta_3\,C_{\opF,\psi}^\star}2
    -\frac{\delta_5\,C_{\opG,s}^\star}2 \right)
    \norm{\Scer_s^k}_{L^2(\Surf_\frc)}^2
    \\[1mm] & \hspace{0.0cm}
    \leq
    \frac{\Dt\,C_S C_{\rho^0}^\star}{8\pi \Grav\,\delta_6}
    \norm{\Veer_\vev^{k-1}}_{L^2(\Domn)}^2
    +\left(\frac{C_{\opF,\sigma}^\star (C_I+C'_I) }
    {2\Dt \delta_7}
    + \frac{C'_I }{2\Dt \delta_8}\right)
    \norm{\Veer_\veu^{k-1}}_{\spH}^2
    +\frac{C_{\opF,\psi}^\star}{2\delta_3}
    \norm{\Scer_\psi^{k-1}}_{L^2(\Surf_\frc)}^2 .
    \end{aligned}
    \label{eq:disc err est}
\end{equation}
With $C_{\opF,\psi}^\star\delta_3=C_{\opG,s}^\star\delta_5=C_\veps
C_{\opF,s}$, $\delta_6=\lambda^{-1/2}$ and
$C_{\opF,\sigma}^\star(C_I+C'_I)\delta_7=C'_I\delta_8=\gamma/\Dt$,
(\ref{eq:disc err est}) becomes
\begin{equation}
    \begin{aligned}
        &
    \left(\frac{C_{\rho^0}}\Dt
    -\frac{\Dt C_S C_{\rho^0}^\star}{8\,\pi \Grav\sqrt{\lambda}}
    \right)
    \norm{\Veer_\vev^k}_{L^2(\Domn)}^2
    +\frac1\Dt C_{a_3}\norm{\Veer_\veu^k}_{\spH}^2
    \\[1mm] & \hspace{1.0cm}
    +\left(\frac1{\Dt}
    -\frac{C_{\opG,s}^{\star\,2}}{2C_\veps C_{\opF,s}}
    +C_{\opG,\psi}\right)
    \norm{\Scer_\psi^k}_{L^2(\Surf_\frc)}^2
    \\[1mm] & \hspace{0.0cm}
    \leq
    \frac{\Dt\,C_S C_{\rho^0}^\star\sqrt{\lambda}}{8\pi \Grav}
    \norm{\Veer_\vev^{k-1}}_{L^2(\Domn)}^2
    +\frac1{2\gamma}\big(C_{\opF,\sigma}^{\star\,2} (C_I+C'_I)^2
    + C^{'\,2}_I \big)
    \norm{\Veer_\veu^{k-1}}_{\spH}^2
    \\[1mm] & \hspace{1.0cm}
    +\frac{C_{\opF,\psi}^{\star\,2}}{2C_\veps C_{\opF,s}}
    \norm{\Scer_\psi^{k-1}}_{L^2(\Surf_\frc)}^2.
    \end{aligned}
    \label{eq:disc err est 1}
\end{equation}
Clearly, the contraction \eqref{eq:contraction 2} is a direct
consequence of (\ref{eq:disc err est 1}) if the criteria in
(\ref{eq:disc contr cond}) are satisfied and $\Dt$ is sufficiently
small. We therefore obtain a unique fixed point
$(\hat\vev^\br{n},\hat\veu^\br{n},\hat\psi^\br{n})^T \in \hat
V_1\times \hat V_1\times \hat V_2$ that solves the time-discretized
problem.
\end{proof}

\begin{remark}
    Theorem \ref{th2} indicates that $\gamma$ can be chosen
    proportional to $\Dt$ to ensure that the general time-discretized
    coupling problem converges to a unique solution. This result is
    consistent with Theorem \ref{th1}. Moreover, the choice of the
    ratio $\gamma/\Dt$ involves the constant
    $C^{*\,2}_{\mathcal{F},\sigma}(C_I+C_I')^2+C_I^{'\,2}$, which is
    related to the smoothness of rupture surface $\Surf_\frc$ as well
    as to the instantaneous friction coefficient.  In some numerical
    tests for ruptures with simple geometry, this constant is small,
    and the criteria for artificial viscosity can be fulfilled by the
    dissipative nature of numerical schemes.  However, as the
    numerical experiments in our companion publication show
    \cite{Ye2018ruptnum}, when the rupture surfaces are nonplanar or
    the elastic material is distinct across the rupture, a
    sufficiently large positive artificial viscosity is necessary to
    guarantee the scheme's convergence.
    \label{rmk th2}
\end{remark}

\section{Discussion}
\label{sec:conclusion}

We establish a mathematical understanding of coupling spontaneous
rupturing and seismic wave generation in a self-gravitating earth by
developing an iterative scheme. We introduce an artificial viscosity
term as a regularization in the relevant elastic-gravitational system
of equations and show the contraction of the iterative scheme in
natural norms. Thus we obtain framework for studying earthquakes with
general rate- and state-dependent friction laws constrained by
observations from experiments. We also give precise conditions on the
viscosity coefficient and time step that guarantee the convergence of
the iterative scheme.

Our iterative coupling scheme provides a natural multi-rate time
stepping strategy for dealing with the nonlinearity of the ordinary
differential equation for state evolution. This evolution requires a
significantly finer time step than the seismic wave propagation and
scattering. We also provide an analysis for the discrete time problem.
Our analysis is carried out with a uniform time step, but the
extension to the multirate case can be made, which is illustrated in a
companion paper \cite{Ye2018ruptnum}.

\section*{Acknowledgement}

R. Ye acknowledges the support from the Simons Foundation under the
MATH $+$ X program and by the members of the Geo-Mathematical Imaging
Group at Rice University. K. Kumar acknowledges Toppforsk, Norwegian
Research Council project 250223. M.V. de Hoop acknowledges the support
from the Simons Foundation under the MATH $+$ X program, the National
Science Foundation under grant DMS-1559587 and the members of the
Geo-Mathematical Group at Rice University.  All the authors
acknowledge Dr. A. Mazzucato at Penn State University for invaluable
suggestions in the preparation of this manuscript.

\bibliographystyle{amsplain}
\bibliography{Rupture} 

\providecommand{\bysame}{\leavevmode\hbox to3em{\hrulefill}\thinspace}
\providecommand{\MR}{\relax\ifhmode\unskip\space\fi MR }
\providecommand{\MRhref}[2]{%
  \href{http://www.ams.org/mathscinet-getitem?mr=#1}{#2}
}
\providecommand{\href}[2]{#2}
\begin{thebibliography}{10}

\bibitem{Aagaard2013}
Brad~T Aagaard, Matthew~G Knepley, and Charles~A Williams, \emph{A domain
  decomposition approach to implementing fault slip in finite-element models of
  quasi-static and dynamic crustal deformation}, Journal of Geophysical
  Research: Solid Earth \textbf{118} (2013), no.~6, 3059--3079.

\bibitem{Andrews2002}
David~J Andrews, \emph{A fault constitutive relation accounting for thermal
  pressurization of pore fluid}, Journal of Geophysical Research: Solid Earth
  \textbf{107} (2002), no.~B12, ESE--15.

\bibitem{Benjemaa2009}
Mondher Benjemaa, Nathalie Glinsky-Olivier, V{\'\i}ctor~M Cruz-Atienza, and
  Jean Virieux, \emph{3-d dynamic rupture simulations by a finite volume
  method}, Geophysical Journal International \textbf{178} (2009), no.~1,
  541--560.

\bibitem{Beretta2017}
Elena Beretta, Maarten~V de~Hoop, Elisa Francini, Sergio Vessella, and Jian
  Zhai, \emph{Uniqueness and lipschitz stability of an inverse boundary value
  problem for time-harmonic elastic waves}, Inverse Problems \textbf{33}
  (2017), no.~3, 035013.

\bibitem{Beretta2014}
Elena Beretta, Maarten~V de~Hoop, Lingyun Qiu, and Otmar Scherzer,
  \emph{Inverse boundary value problem for the helmholtz equation: multi-level
  approach and iterative reconstruction}, arXiv preprint arXiv:1406.2391
  (2014), 1--20.

\bibitem{Brazda2017}
Katharina Brazda, Maarten~V de~Hoop, and Guenther Hoermann, \emph{Variational
  formulation of the earth's elastic-gravitational deformations under low
  regularity conditions}, arXiv preprint arXiv:1702.04741 (2017), 1--81.

\bibitem{Bureau2000}
Lionel Bureau, Tristan Baumberger, and Christiane Caroli, \emph{Shear response
  of a frictional interface to a normal load modulation}, Physical Review E
  \textbf{62} (2000), no.~5, 6810.

\bibitem{Dahlen1972}
Francis~A Dahlen, \emph{Elastic dislocation theory for a self-gravitating
  elastic configuration with an initial static stress field}, Geophysical
  Journal International \textbf{28} (1972), no.~4, 357--383.

\bibitem{Dahlen1973}
\bysame, \emph{Elastic dislocation theory for a self-gravitating elastic
  configuration with an initial static stress field ii. energy release},
  Geophysical Journal International \textbf{31} (1973), no.~4, 469--484.

\bibitem{Dahlen1977}
\bysame, \emph{The balance of energy in earthquake faulting}, Geophysical
  Journal International \textbf{48} (1977), no.~2, 239--261.

\bibitem{Dalguer2007}
Luis~A Dalguer and Steven~M Day, \emph{Staggered-grid split-node method for
  spontaneous rupture simulation}, Journal of Geophysical Research: Solid Earth
  \textbf{112} (2007), no.~B2, 1--15.

\bibitem{Day2005}
Steven~M Day, Luis~A Dalguer, Nadia Lapusta, and Yi~Liu, \emph{Comparison of
  finite difference and boundary integral solutions to three-dimensional
  spontaneous rupture}, Journal of Geophysical Research: Solid Earth
  \textbf{110} (2005), no.~B12, 1--23.

\bibitem{deHoop2015}
Maarten~V de~Hoop, Sean Holman, and Ha~Pham, \emph{On the system of
  elastic-gravitational equations describing the oscillations of the earth},
  arXiv preprint arXiv:1511.03200 (2015), 1--51.

\bibitem{delaPuente2009}
Josep de~la Puente, Jean-Paul Ampuero, and Martin K{\"a}ser, \emph{Dynamic
  rupture modeling on unstructured meshes using a discontinuous galerkin
  method}, Journal of Geophysical Research: Solid Earth \textbf{114} (2009),
  no.~B10, 1--17.

\bibitem{Dieterich1979}
James~H Dieterich, \emph{Modeling of rock friction: 1. experimental results and
  constitutive equations}, Journal of Geophysical Research: Solid Earth
  \textbf{84} (1979), no.~B5, 2161--2168.

\bibitem{Dunham2011}
Eric~M Dunham, David Belanger, Lin Cong, and Jeremy~E Kozdon, \emph{Earthquake
  ruptures with strongly rate-weakening friction and off-fault plasticity, part
  1: Planar faults}, Bulletin of the Seismological Society of America
  \textbf{101} (2011), no.~5, 2296--2307.

\bibitem{Duru2016}
Kenneth Duru and Eric~M Dunham, \emph{Dynamic earthquake rupture simulations on
  nonplanar faults embedded in 3d geometrically complex, heterogeneous elastic
  solids}, Journal of Computational Physics \textbf{305} (2016), 185--207.

\bibitem{Geubelle1995}
Philippe~H Geubelle and James~R Rice, \emph{A spectral method for
  three-dimensional elastodynamic fracture problems}, Journal of the Mechanics
  and Physics of Solids \textbf{43} (1995), no.~11, 1791--1824.

\bibitem{Girault2016}
Vivette Girault, Kundan Kumar, and Mary~F Wheeler, \emph{Convergence of
  iterative coupling of geomechanics with flow in a fractured poroelastic
  medium}, Computational Geosciences \textbf{20} (2016), no.~5, 997--1011.

\bibitem{Harris2017}
Ruth~A Harris, \emph{Large earthquakes and creeping faults}, Reviews of
  Geophysics \textbf{55} (2017), no.~1, 169--198.

\bibitem{Ionescu2002}
Ioan~R Ionescu, \emph{Viscosity solutions for dynamic problems with slip-rate
  dependent friction}, Quarterly of Applied Mathematics \textbf{60} (2002),
  no.~3, 461--476.

\bibitem{Ionescu2003}
Ioan~R Ionescu, Quoc-Lan Nguyen, and Sylvie Wolf, \emph{Slip-dependent friction
  in dynamic elasticity}, Nonlinear Analysis: Theory, Methods \& Applications
  \textbf{53} (2003), no.~3, 375--390.

\bibitem{Ionescu1996}
Ioan~R Ionescu and Jean-Claude Paumier, \emph{On the contact problem with slip
  displacement dependent friction in elastostatics}, International journal of
  engineering science \textbf{34} (1996), no.~4, 471--491.

\bibitem{Kaneko2008a}
Yoshihiro Kaneko, Nadia Lapusta, and Jean-Paul Ampuero, \emph{Spectral element
  modeling of spontaneous earthquake rupture on rate and state faults: Effect
  of velocity-strengthening friction at shallow depths}, Journal of Geophysical
  Research: Solid Earth \textbf{113} (2008), no.~B9, 1--17.

\bibitem{Klarbring1988}
Anders Klarbring, Andro Mikeli{\'c}, and Meir Shillor, \emph{Frictional contact
  problems with normal compliance}, International Journal of Engineering
  Science \textbf{26} (1988), no.~8, 811--832.

\bibitem{Kozdon2013}
Jeremy~E Kozdon, Eric~M Dunham, and Jan Nordstr{\"o}m, \emph{Simulation of
  dynamic earthquake ruptures in complex geometries using high-order finite
  difference methods}, Journal of Scientific Computing \textbf{55} (2013),
  no.~1, 92--124.

\bibitem{Linker1992}
M.~F. Linker and James~H Dieterich, \emph{Effects of variable normal stress on
  rock friction: Observations and constitutive equations}, Journal of
  Geophysical Research: Solid Earth \textbf{97} (1992), no.~B4, 4923--4940.

\bibitem{Lozos2015}
Julian~C Lozos, Ruth~A Harris, Jessica~R Murray, and James~J Lienkaemper,
  \emph{Dynamic rupture models of earthquakes on the bartlett springs fault,
  northern california}, Geophysical Research Letters \textbf{42} (2015),
  no.~11, 4343--4349.

\bibitem{Martins1987}
Jo\~ao Arm\'enio~Correia Martins and J~Tinsley Oden, \emph{Existence and
  uniqueness results for dynamic contact problems with nonlinear normal and
  friction interface laws}, Nonlinear Analysis: Theory, Methods \& Applications
  \textbf{11} (1987), no.~3, 407--428.

\bibitem{OReilly2015}
Ossian O'Reilly, Jan Nordstr{\"o}m, Jeremy~E Kozdon, and Eric~M Dunham,
  \emph{Simulation of earthquake rupture dynamics in complex geometries using
  coupled finite difference and finite volume methods}, Communications in
  Computational Physics \textbf{17} (2015), no.~02, 337--370.

\bibitem{Pelties2012}
Christian Pelties, Josep Puente, Jean-Paul Ampuero, Gilbert~B Brietzke, and
  Martin K{\"a}ser, \emph{Three-dimensional dynamic rupture simulation with a
  high-order discontinuous galerkin method on unstructured tetrahedral meshes},
  Journal of Geophysical Research: Solid Earth \textbf{117} (2012), no.~B2,
  1--15.

\bibitem{Pipping2017}
Elias Pipping, \emph{Existence of long-time solutions to dynamic problems of
  viscoelasticity with rate-and-state friction}, arXiv preprint
  arXiv:1703.04289 (2017), 1--12.

\bibitem{Pipping2015}
Elias Pipping, Oliver Sander, and Ralf Kornhuber, \emph{Variational formulation
  of rate-and state-dependent friction problems}, ZAMM-Journal of Applied
  Mathematics and Mechanics/Zeitschrift f{\"u}r Angewandte Mathematik und
  Mechanik \textbf{95} (2015), no.~4, 377--395.

\bibitem{Prakash1998}
Vikas Prakash, \emph{Frictional response of sliding interfaces subjected to
  time varying normal pressures}, Journal of Tribology \textbf{120} (1998),
  no.~1, 97--102.

\bibitem{Quarteroni2008}
Alfio Quarteroni and Alberto Valli, \emph{Numerical approximation of partial
  differential equations}, vol.~23, Springer Science \& Business Media, 2008.

\bibitem{Rice1983}
James~R Rice, \emph{Constitutive relations for fault slip and earthquake
  instabilities}, Pure and applied geophysics \textbf{121} (1983), no.~3,
  443--475.

\bibitem{Rice2006}
\bysame, \emph{Heating and weakening of faults during earthquake slip}, Journal
  of Geophysical Research: Solid Earth \textbf{111} (2006), no.~B5, 1--29.

\bibitem{Rice2001}
James~R Rice, Nadia Lapusta, and K~Ranjith, \emph{Rate and state dependent
  friction and the stability of sliding between elastically deformable solids},
  Journal of the Mechanics and Physics of Solids \textbf{49} (2001), no.~9,
  1865--1898.

\bibitem{Rice1983a}
James~R Rice and Andy~L Ruina, \emph{Stability of steady frictional slipping},
  Journal of applied mechanics \textbf{50} (1983), no.~2, 343--349.

\bibitem{Richardson1999}
Eliza Richardson and Chris Marone, \emph{Effects of normal stress vibrations on
  frictional healing}, J. geophys. Res \textbf{104} (1999), no.~B12,
  28859--28878.

\bibitem{Ruina1983}
Andy Ruina, \emph{Slip instability and state variable friction laws}, Journal
  of Geophysical Research: Solid Earth \textbf{88} (1983), no.~B12,
  10359--10370.

\bibitem{Ruina1981}
Andy~L Ruina, \emph{Friction laws and instabilities: A quasistatic analysis of
  some dry frictional behavior}, Ph.D. thesis, Brown University, 1981.

\bibitem{Salo2008}
Mikko Salo, \emph{Calder{\'o}n problem}, 2008.

\bibitem{Schmitt2015}
Stuart~V Schmitt, Paul Segall, and Eric~M Dunham, \emph{Nucleation and dynamic
  rupture on weakly stressed faults sustained by thermal pressurization},
  Journal of Geophysical Research: Solid Earth \textbf{120} (2015), no.~11,
  7606--7640.

\bibitem{Segall2015}
Paul Segall and San Lu, \emph{Injection-induced seismicity: Poroelastic and
  earthquake nucleation effects}, Journal of Geophysical Research: Solid Earth
  \textbf{120} (2015), no.~7, 5082--5103.

\bibitem{Shapiro2009}
Serge~A Shapiro and Carsten Dinske, \emph{Fluid-induced seismicity: Pressure
  diffusion and hydraulic fracturing}, Geophysical Prospecting \textbf{57}
  (2009), no.~2, 301--310.

\bibitem{Shi2018}
Jia Shi, Ruipeng Li, Yuanzhe Xi, Yousef Saad, and Maarten~V. de~Hoop,
  \emph{Computing planetary interior normal modes with a highly parallel
  polynomial filtering eigensolver}, Proceedings of International Conference
  for High Performance Computing, Networking, Storage and Analysis (SC18),
  Dallas, TX, Nov. 11-16, 2018.

\bibitem{Tago2012}
Josu{\'e} Tago, V{\'\i}ctor~M Cruz-Atienza, Jean Virieux, Vincent Etienne, and
  Francisco~J S{\'a}nchez-Sesma, \emph{A 3d hp-adaptive discontinuous galerkin
  method for modeling earthquake dynamics}, Journal of Geophysical Research:
  Solid Earth \textbf{117} (2012), no.~B9, 1--21.

\bibitem{Thomas2017}
Marion~Y. Thomas, Jean-Philippe Avouac, and Nadia Lapusta, \emph{Rate-and-state
  friction properties of the longitudinal valley fault from kinematic and
  dynamic modeling of seismic and aseismic slip}, Journal of Geophysical
  Research: Solid Earth \textbf{122} (2017), no.~4, 3115--3137, 2016JB013615.

\bibitem{Valette1987}
Bernard Valette, \emph{About the influence of pre-stress upon adiabatic
  perturbations of the earth}, Geophysical Journal International \textbf{85}
  (1986), no.~1, 179--208.

\bibitem{Woodhouse1978}
Jim~H Woodhouse and Francis~A Dahlen, \emph{The effect of a general aspherical
  perturbation on the free oscillations of the earth}, Geophysical Journal
  International \textbf{53} (1978), no.~2, 335--354.

\bibitem{Ye2018ruptnum}
Ruichao Ye, Kundan Kunmar, Maarten~V. de~Hoop, and Michel Campillo, \emph{A
  multi-rate iterative coupling scheme for simulating dynamic ruptures and
  seismic waves generation in the self-gravitating earth}, Journal of
  Computational Physics (2019), In print.

\bibitem{Zhang2014}
Zhenguo Zhang, Wei Zhang, and Xiaofei Chen, \emph{Three-dimensional curved grid
  finite-difference modelling for non-planar rupture dynamics}, Geophysical
  Journal International \textbf{199} (2014), no.~2, 860--879.

\end{thebibliography}

\end{document}